%% file: main-ieee.tex
\documentclass[journal,10pt]{IEEEtran}

\usepackage[numbers,sort&compress]{natbib}

\usepackage{booktabs} 
\usepackage[linesnumbered,commentsnumbered,ruled,vlined]{algorithm2e}

\SetCommentSty{mycommfont}

\SetKwRepeat{Do}{do}{while}

\usepackage{ulem}
\usepackage{multicol}
\usepackage{latexsym}
\usepackage{multirow}
\usepackage{graphicx}
\usepackage{threeparttable}
\usepackage{amsmath}
\usepackage{amsfonts}
\usepackage{pifont}
\usepackage{verbatim}
\usepackage{url}
\usepackage{hyperref}
\hypersetup{
	colorlinks=true,
	linkcolor=blue,
	citecolor=violet,
	filecolor=magenta,      
	urlcolor=blue,
}
\usepackage{tcolorbox}
\usepackage{galois}
\usepackage{caption,subcaption}
\usepackage{tikz}
\usepackage{pgfplots}
\pgfplotsset{compat=1.10} 

\usepackage[utf8]{inputenc}
\usepackage{cleveref}
\crefname{section}{§}{§§}
\Crefname{section}{§}{§§}

\newcommand{\mytodoblue}[1]{\textcolor{blue}{\ding{46}~{\sf}~#1}}
\newcommand{\mytodored}[1]{\textcolor{red}{\ding{46}~{\sf}~#1}}

\newcommand{\ToolName}{\text{Spear}}
\newif\ifshowcomments
\showcommentsfalse

\ifshowcomments
\newcommand{\yao}[1]{\mytodored{[yao: #1]}}
\newcommand{\zuo}[1]{\mytodored{[zuo: #1]}}
\newcommand{\wang}[1]{\mytodored{[wang: #1]}}
\newcommand{\li}[1]{\mytodoblue{[li: #1]}}
\newcommand{\revise}[1]{\mytodobluetwo{#1}}
\else
\newcommand{\yao}[1]{}
\newcommand{\zuo}[1]{}
\newcommand{\wang}[1]{}
\newcommand{\li}[1]{}
\newcommand{\revise}[1]{#1}
\fi

\usepackage{amsthm} 
\newtheorem{definition}{Definition} 
\newtheorem{theorem}{Theorem} 
\theoremstyle{definition} 
\newtheorem{example}{Example}[section] 

\usepackage{balance}

\usepackage{fontawesome5}

\begin{document}

\title{Synthesizing Best Abstract Transformers via Parallel Bit-Vector Optimization}


\author{%
Weiqi Wang$^{1}$, Peisen Yao$^{1}$, Hanrui Zuo$^{1}$, Yuan Li$^{1}$, Hongfei Fu$^{2}$, Kui Ren$^{1}$ \\
$^{1}$The State Key Laboratory of Blockchain and Data Security, Zhejiang University, China \\
$^{2}$ Shanghai Jiao Tong University, China \\
Email: \{wqwang1009, pyaoaa, zarin, ly.liyuan, kuiren\}@zju.edu.cn, jt002845@sjtu.edu.cn
}

\maketitle
\IEEEdisplaynontitleabstractindextext

\begin{abstract}
Abstract interpretation provides a principled foundation for constructing sound static analyses through systematic abstraction. A central challenge is synthesizing the best abstract transformers that achieve optimal precision within a given abstract domain. This paper addresses this problem for low-level code modeled with fixed-size bit-vectors. Recent approaches formulate the synthesis task as a multi-objective Optimization Modulo Theories (OMT) problem, but suffer from limited scalability. 
We introduce Spear, a parallel synthesis framework that exploits a key structural insight: while the bits within each objective must be processed sequentially, the objectives themselves are independent. Spear leverages the independence of inter-objective bits to better parallelize the synthesis. Experimental results on benchmarks across two binary analysis domains show that Spear consistently outperforms state-of-the-art OMT solvers, solving more instances and achieving significantly improved runtimes. To our knowledge, this is the first approach to apply parallelism to accelerate the synthesis of optimal abstract transformers.
\end{abstract}



\input{1intro}

\input{2background}
\input{3overview}

\input{4divide}

\input{5evaluation}

\input{6discuss}

\input{7related}
\input{8conclusion}
\input{9dataavail}

\balance
\footnotesize{
\bibliographystyle{unsrtnat}
\bibliography{opt,tmp,psmt}
}

\end{document}

%% file: 1intro.tex
\section{Introduction}
\label{sec:intro}

Abstract interpretation~\cite{cousot1976static,cousot1978automatic} provides a mathematical framework for constructing sound static program analyses. 
The framework is founded on abstract domains that encode over-approximations of concrete program states, with abstract transformers providing a systematic means of approximating program semantics.
\revise{Building an effective analyzer requires both an appropriate abstract domain and efficient transfer of program semantics into that domain.} A fundamental theoretical consideration is the selection of abstract domains that capture the properties of interest, such as numerical invariants~\cite{singh2018fast} or memory shapes~\cite{reps2004symbolic}.

A central theoretical result due to \citet{DBLP:conf/popl/CousotC79} establishes that given a Galois connection between concrete domain $C$ and abstract domain $A$, and concrete transformer $f: C \to C$, there exists an optimal abstract transformer $f^\#: A \to A$ that represents the theoretical limit of precision attainable within the chosen abstraction. As summarized in Table~\ref{tbl:existing_work}, significant research has focused on automatically synthesizing optimal abstract transformers across various abstract domains, including finite-height domains~\cite{reps2004symbolic} and template linear domains~\cite{monniaux2009automatic,brauer2010automatic}. These advances have enabled numerous applications, such as program verification~\cite{reps2004symbolic,li2014symbolic,jiang2017block} \revise{
control-flow recovery~\cite{barrett2010range,Reinbacher:2011}, and binary memory analysis~\cite{VSA:Reps08}}.
\revise{For software-engineering tasks, the practical issue is not merely whether the best transformer exists, but whether it can be synthesized within the time budget of an analysis pipeline that may issue thousands of solver queries.}
\begin{table*}[t]
\centering
\caption{Representative work on best abstract transformer synthesis}
\label{tbl:existing_work}
{
\begin{tabular}{ c  l  l  l    l  }
\toprule
\textbf{Input}    & \textbf{Theory} & \textbf{Abstract Domain}  & \textbf{Applications}    &      \\ \midrule
Source & Real Arithmetic & Octagon &      Program Verification~\cite{li2014symbolic}    \\       
Source & Real  Arithmetic  & Polyhedron &     Program Verification~\cite{jiang2017block}        \\ 
Source & Real Arithmetic & Template polyhedra & Program Verification~\cite{monniaux2009automatic} \\
Source &  Bitvec  Arithmetic  & Interval, bitmask  &       Compiler Optimization~\cite{ritter2015compiler}                \\   
Source &    FOL   & Shape graphs &           Shape analysis~\cite{reps2004symbolic,DBLP:conf/tacas/YorshRS04}       & \\          
Binary &   Bitvec  Arithmetic  & Affine relation         & Binary Analysis ~\cite{elder2014abstract}       \\  
Binary &   Bitvec  Arithmetic  & Range, set       &   Binary Analysis ~\cite{brauer2010automatic,brauer2011transfer}       \\  
Binary &  Bitvec  Arithmetic  & Polyhedron    &  Bug detection~\cite{yao2021program,huang2020pangolin}          \\ \bottomrule
\end{tabular}
}	
\end{table*}

This work addresses the synthesis of abstract transformers for low-level programs, such as X86 assembly, in which machine integers are encoded as fixed-size bit-vectors~\cite {barrett2010range,poeplau2020symbolic}.
\revise{This setting is central to binary analysis, where analyzers must recover indirect control flow, reason about registers and memory, and support vulnerability-discovery workflows over stripped executables~\cite{shoshitaishvili2016sok,Cha:2012:UMB:2310656.2310692,ye2024manta,zhou2024plankton}.}
The significance of abstract transformer synthesis in this domain is twofold: it mitigates precision loss in analyzing sequences of instructions through holistic transformer synthesis, as the optimality of abstract transformers is not closed under composition~\cite{cousot1978automatic};  it provides a rigorous and automated treatment of machine arithmetic semantics such as wrap-around effects, which would otherwise be tedious and error-prone to model manually~\cite{tsl,thakur2012bilateral}.

To address the problem of synthesizing the best abstract transformers, recent approaches have leveraged Optimization Modulo Theories (OMT)~\cite{tsiskaridze2024generalized,li2014symbolic}, a framework that extends Satisfiability Modulo Theories (SMT) with optimization capabilities. Specifically, the problem of abstract transformer synthesis is formulated as solving \emph{boxed, multi-objective OMT} instances of the form:
$$\text{max } \{ g_1, \ldots, g_n \} \text{ s.t.}  \varphi,$$ where $g_1, \ldots, g_n$ \revise{are separate optimization targets} and $\varphi$ represents the constraining SMT formula.

In our formulation, $\varphi$ is expressed in quantifier-free bit-vector logic, with each $g_i$ representing a bit-vector term. 
A key consideration in leveraging OMT solvers for optimal abstract transformer synthesis is their time efficiency. 
\revise{
Existing approaches to OMT(BV) include reductions to alternative reasoning frameworks, such as weighted MaxSAT~\cite{bjorner2015nuz,nadel2016bit} and quantifier instantiation~\cite{brauer2011transfer,kong2018delta-decision}, and iterative SMT-based methods that use binary or linear search~\cite{li2014symbolic,yao2021program,DBLP:journals/corr/abs-1905-02838}.
Modern boxed OMT solvers additionally interleave Boolean search, theory minimization, and improvement clauses, thereby improving several objectives during the same global search process~\cite{sebastiani2015pushing,DBLP:journals/corr/SebastianiT17}. While these approaches have shown promise, scaling them to efficiently handle large and complex constraints remains a significant challenge.}

To accelerate OMT solving, we introduce \ToolName, which enhances time efficiency through parallel computation. \ToolName\ leverages the inherent independence of bits across distinct objectives, enabling fine-grained parallelism by concurrently grouping and processing individual bits.
We evaluate \ToolName\ using benchmarks derived from two binary analysis clients, comparing it against two state-of-the-art OMT solvers: $\nu$Z~\cite{bjorner2015nuz} and OptiMathSAT~\cite{sebastiani2020optimathsat}. The results demonstrate that \ToolName\ consistently outperforms its competitors, solving significantly more instances and achieving superior runtime performance. For example, with eight cores, \ToolName\ solves 45\%–52\% more instances and delivers average speedups of 7.4$\times$ and 5.4$\times$ relative to $\nu$Z and OptiMathSAT, respectively.
To our knowledge, this is the first work to exploit parallelism to accelerate the synthesis of optimal abstract transformers. To summarize, this paper  makes the following contributions:
\begin{itemize}
    \item We introduce \ToolName, the first parallel OMT(BV) framework for synthesizing the best abstract transformers for low-level programs, exploiting \revise{boxed-objective structure, and bit-level scheduling} to enable fine-grained parallelism.
   \item We evaluate \ToolName\ on benchmarks from two binary analysis applications, \revise{namely, hybrid fuzzing and static vulnerability detection,} showing that it solves up to 52\% more instances and achieves speedups of up to 7.4$\times$ over state-of-the-art OMT solvers.
\end{itemize}

The rest of this paper is organized as follows. \Cref{sec:background} reviews the best abstract transformer synthesis and boxed OMT for bit-vectors. \Cref{sec:overview} outlines the single-objective algorithm and states the problem. \Cref{sec:divide-and-conquer} presents our parallel approach. \revise{\Cref{sec:applications} describes the implementation and binary-analysis applications.} \Cref{sec:eval} reports the evaluation; \Cref{sec:discuss} and \Cref{sec:related} discuss implications and related work. \Cref{sec:conclu} concludes.

%% file: 2background.tex
\section{Background and Motivation}
\label{sec:background}

This section provides the basic background for the problem of best abstract transformer synthesis and optimization modulo theory for bit-vector arithmetic.


\subsection{Best Abstract Transformer Synthesis}
Abstract interpretation is a foundational framework for static program analysis, typically formalized via Galois connections between concrete and abstract domains. 
Let $(C, \leq_C)$ and $(A, \leq_A)$ be complete lattices representing the concrete and abstract domains, respectively. The functions $\alpha: C \to A$ and $\gamma: A \to C$ denote abstraction and concretization. The pair $(\alpha, \gamma)$ forms a Galois connection if $\alpha(c) \leq_A a \iff c \leq_C \gamma(a)$ for all $c \in C$ and $a \in A$. 
Given a \textit{concrete transformer} $f : C \to C$, we say an \textit{abstract transformer} $f^\# : A \to A$ is  a \textit{sound abstraction} of $f$ if $\alpha(f(c)) \leq_A f^\#(\alpha(c))$ for any $c \in C$.

\begin{definition} [Best Abstract Transformer]
Given a concrete transformer  $f : C \to C$, the \emph{best (most precise) abstract transformer}  $f^\alpha : A \to A$ that over-approximates  $f$ is
$f^\alpha = \alpha \circ f \circ \gamma : A \to A$, because for any sound abstraction  $f^\#$ it holds that $f^\alpha(a) \leq_A f^\#(a)$ for any abstract value $a \in A$.
\end{definition}

This definition guarantees that $f^\alpha$ is the most precise transformer that can be derived in the abstract domain, thereby establishing a theoretical upper bound on precision.
\footnote{There may be no best abstract transformers~\cite{cousot1978automatic,cousot1995formal,cousot2011logical}. Then, only half of the Galois connection can be used, e.g., the one with concretization function only~\cite{cousot1978automatic}. In this work, we do not account for such abstract domains.}
However, the equation is \textit{non-constructive}, meaning it does not necessarily yield an algorithm for applying $f^\alpha$ or finding a representation of  $f^\alpha$.
Besides, the optimality of abstract transformers is \textit{not closed under composition}, implying that the composition of the best abstractions of two functions $f$ and $g$ may not yield the best abstraction of their composition $f \circ g$.

\smallskip
\noindent \textbf{Synthesizing Best Abstract Transformers}.
To address these challenges, automated synthesis methods have been developed within the symbolic abstraction framework~\cite{reps2004symbolic}.
The core idea is to represent the concrete semantics of a program using a logic formula $\varphi$ in logic $\mathcal{L}$ and then over-approximate the formula with elements in an abstract domain $A$.

\begin{definition} [Symbolic Abstraction]
   \revise{Let $\left[\!\left[ \varphi \right]\!\right]$ denote the set of concrete states that satisfy a formula $\varphi \in \mathcal{L}$. Given $\varphi$ and an abstract domain $A$, the symbolic abstraction problem is to compute the least abstract element $a \in A$ such that $\left[\!\left[ \varphi \right]\!\right] \subseteq \gamma(a)$. Equivalently, symbolic abstraction computes the strongest logical consequence of $\varphi$ expressible in $A$.}
\end{definition}

Depending on the context, the formula $\varphi$ may encode the concrete semantics of different program constructs,  such as a single instruction, a basic block, or a loop-free program fragment.
The problem has been studied for different abstract domains, such as finite-height domains~\cite{reps2004symbolic},
template linear domains~\cite{monniaux2009automatic,brauer2010automatic}, 
and the polyhedral domains~\cite{thakur2012method,yao2021program}.
 
\subsection{Abstract Transformer Synthesis via Bit-Vector Optimization}
Our work is motivated by the need for effective analysis of low-level programs. 
Specifically, we utilize bit-precise static analysis, which models machine integers as fixed-size bit-vectors~\cite {DBLP:conf/issta/JiaH00MZ23} rather than unbounded mathematical integers or reals. In this context, symbolic abstraction aids in two critical aspects of constructing abstract transformers:
\begin{itemize}
    \item First, it mitigates precision loss arising from the fact that optimal abstract transformers are not closed under composition~\cite{cousot1978automatic}.
    Consider a sequence of $n$ concrete operations $g_1, \ldots, g_n$, each approximated by an abstract transformer $f_1, \ldots, f_n$. Applying these transformers in sequence to an abstract input $d$ yields the result $f_n(\ldots f_2(f_1(d)))$. However, greater precision can be achieved by synthesizing a single abstract transformer that directly approximates the overall effect of the composition $g_n \circ \cdots \circ g_1$. 
    \item  Second, it allows the automatic synthesis of abstract transformers for various low-level operations,
    which would otherwise be tedious and error-prone~\cite{tsl,sharma2017sound}. 
     For instance, consider an operation that increments a value by one. In machine arithmetic, if the increment is applied to the largest representable integer, the result wraps around to the smallest representable integer. The abstract transformer must accurately capture this behavior~\cite{tsl,sharma2017sound}.
\end{itemize}

\begin{example}
Consider the following x86 assembler code for a switch table:
\begin{verbatim}
   mov eax, [ebp-0x8] ; eax := *(ebp - 8)
   sub eax, 0x2       ; eax := eax - 2
   cmp eax, 0x5       ; CF := (0 <= eax < 5)
                      ; ZF := (eax == 5)
   ja 0xd8             
   jmp [0x8048a0c + eax*4]
\end{verbatim}

To construct the control flow graph (CFG), we need to analyze the value of \texttt{eax} (e.g., $eax \in [0, 5]$) when the indirect jump is executed. This requires careful reasoning about the carry (CF) and zero (ZF) flags, which are set by the \texttt{cmp} instruction. Symbolic abstraction facilitates precise reasoning about various instructions and their impact on control flow.
\end{example}

This \revise{\sout{works} work} targets \emph{template linear constraint domains}~\cite{colon2003linear}, where an element $F$ is a conjunction of linear inequalities $L_1(s_1, \dots, s_m) \leq p_1 \land \dots \land L_n(s_1, \dots, s_m) \leq p_n$.
Here, the left-hand sides $L_1, \dots, L_n$ are fixed linear forms over variables $s_1, \dots, s_m$, and the right-hand sides $p_1, \dots, p_n$ are parameters.
Particular examples of abstract domains in this class are:
\begin{itemize}
\item Interval domain: For each variable $s$, consider the linear forms $s$ and $-s$.
\item Zone domain: For each pair of variables $s_1$ and $s_2$, consider the linear form $s_1 - s_2$.
\item Octagon domain: For each pair of variables $s_1$ and $s_2$, consider the linear forms $\pm s_1 \pm s_2$.
\end{itemize}

For such domains, symbolic abstraction reduces to solving \emph{boxed} optimization modulo theories (OMT) problems~\cite{sebastiani2015pushing}.

\begin{definition} [Boxed OMT Solving]
  Given an SMT formula $\varphi$ and a set of objectives ${ g_1, \ldots, g_n }$, the goal of the \emph{multiple-independent-objective OMT problem}, also known as \emph{boxed OMT}~\cite{sebastiani2015pushing}, is to find a set of models ${ M_1, \ldots, M_n }$ of $\varphi$ such that each $M_i$ maximizes the respective objective $g_i$. \revise{The objectives are ``independent'' \emph{problems}: each target is optimized separately and may be witnessed by a different model of the same formula. They are not independent variables because, obviously, all objectives remain constrained by the shared formula $\varphi$.}
  \footnote{It is important to distinguish boxed OMT from conventional multi-objective optimization techniques, such as Pareto optimization (where objectives interfere with each other) and lexicographic optimization (where objectives are prioritized in strict order).} 
\end{definition}

To illustrate, consider the interval domain abstraction for variables \( x \) and \( y \), represented by the constraints \( \{ m_1 \leq x, x \leq M_1, m_2 \leq y, y \leq M_2 \} \), where \( m_1, M_1, m_2, M_2 \) are unknown parameters to be synthesized. These parameters can be computed by solving one boxed OMT problem:  
\[
\text{max } \{ x, y, -x, -y \} \ \text{s.t.} \ \varphi.
\]  

\revise{
The following theorem states the precision property used by this reduction.}

\revise{
\begin{theorem}[Precision of interval symbolic abstraction]
\label{thm:interval-precision}
Let $\varphi$ be a satisfiable formula over variables $x_1,\ldots,x_m$, and let $u_j=\max x_j$ and $\ell_j=-\max(-x_j)$ be the optima obtained by solving the boxed OMT instance with objectives $\{x_1,-x_1,\ldots,x_m,-x_m\}$ subject to $\varphi$. Then the interval element $\bigwedge_j \ell_j \leq x_j \leq u_j$ is the least interval abstraction that contains $\left[\!\left[ \varphi \right]\!\right]$.
\end{theorem}}

\revise{
\noindent\textit{Proof sketch.}
Soundness follows directly from optimality: for every state $\sigma \in \left[\!\left[ \varphi \right]\!\right]$, $\sigma(x_j) \leq u_j$ and $-\sigma(x_j) \leq -\ell_j$, hence $\ell_j \leq \sigma(x_j) \leq u_j$. Minimality follows because any sound interval abstraction must contain every satisfying value of each $x_j$; therefore, its upper bound for $x_j$ cannot be smaller than $u_j$, and its lower bound cannot be larger than $\ell_j$. Thus no strictly smaller interval element can over-approximate $\left[\!\left[ \varphi \right]\!\right]$.}

\begin{example}
Consider the SMT formula $\varphi(x, y) \equiv x \geq 0 \land y \geq 0 \land x + y \leq 5$, which defines a region where $x$ and $y$ are non-negative, and their sum is bounded by 5. Using the template $\{ x, y, -x, -y \}$, we solve the boxed OMT problem $\text{max } \{ x, y, -x, -y \} \ \text{s.t.} \ \varphi.$
The solution yields the maximum and minimum values of $x$ and $y$, resulting in the precise interval abstraction $x \in [0, 5] \land y \in [0, 5]$.
\end{example}

\revise{In the binary-analysis clients evaluated in this paper, the constraint formula $\varphi$ can arise from path conditions, instruction semantics, register and memory expressions, and control-flow guards.} While boxed OMT provides a general framework for \revise{
these} problems, the computational overhead can be significant for large and complex formulas. 
To mitigate this, we propose parallel solving strategies 
to enhance scalability for real-world applications.


%% file: 3overview.tex
\section{Overview}
\label{sec:overview}
In this section, we first recall a basic algorithm for optimizing a single objective, and then discuss parallel optimization of multiple objectives. To simplify the presentation, we assume that all the bit-vectors discussed are unsigned.

\subsection{Optimizing One Objective}

\begin{algorithm}[t]
	\caption{OBV-BS for solving single-objective OMT(BV) problems}
         \label{alg:basic}
	\KwIn{A quantifier-free bit-vector formula  $\varphi$ and an objective $g$}
	\KwOut{An optimal model of $\varphi$ that maximizes the objective function $g$}
 	\SetKwFunction{obv-bs}{OBV-BS}
	\SetKwProg{Fn}{Function}{}{}
	\Fn{OBV-BS{($\varphi, g$)}}{
 	Boolean formula  $\phi$ $\leftarrow$ bit-blast the bit-vector formula $\varphi$\;
	Boolean variables $T = \{t_{n-1}, \ldots, t_0\}$ $\leftarrow$ bit-blast the objective $g$\;
	$S \leftarrow Solver(), S.add(\phi)$\;
 	$lits \leftarrow \{ \}$\tcp*{the assumption literals tracking the progress}
        \label{line:lits}
	$S.check(), M \leftarrow S.get\_model()$\;
	\For{$i \leftarrow n -1$ {\rm \textbf{downto}} 0 {\rm \textbf{step}} 1}{
		\tcp*{examine whether $t_i$ can be 1 until finding the maximum value of $g$, i.e., maximizing $f(t_{n-1}, \ldots, t_0) = 2^{n-1}*t_{n-1} + \cdots + 2^0 * t_0$}
		\If{$M.eval(t_i) == true$}{ 
                \label{line:check_model_ti}
			$lits \leftarrow lits \cup \{ t_i\}$\tcp*{update the assumptions using $t_i = 1$}
                 \label{line:check_model_success}
		} \Else{
			\If{$S.check(lits \cup \{t_i \}) == SAT$} {
                    \label{line:check_next_bit}
				$M \leftarrow S.get\_model()$\;
                      $lits \leftarrow lits \cup \{t_i\}$\tcp*{update the assumptions using $t_i = 1$}
			} \Else{
				$lits \leftarrow lits \cup \{ \neg t_i\}$\tcp*{update the assumptions using $t_i = 0$}
                 \label{line:set_ti_zero}
			}
		}
	}
	\Return $M$\; 
	}
\end{algorithm}
We begin by outlining the basic algorithm for optimizing a single objective, proposed by \citet{nadel2016bit} (denoted the OBV-BS algorithm).
The goal is to maximize an $n$-bit objective with respect to a bit-vector formula $\varphi$. The key insight behind OBV-BS is to reduce the OMT(BV) problem, i.e., ``$\text{max} \ g \ \text{s.t.} \ \varphi$'' to a weighted MaxSAT instance through bit-blasting, and then employ a SAT-based procedure for optimization.

\smallskip
\noindent \textbf{Constraint Encoding}.
Given a formula $\varphi$ and an $n$-bit optimization target $g^{[n]}$, we encode $\varphi$ as a hard Boolean formula $\phi$ via bit-blasting~\cite{barrett2018satisfiability}. The objective $g$ is encoded as a set of Boolean variables $T = {t_{n - 1}, \ldots, t_0}$, where $t_i$ represents the $i$-th bit of the bit-vector term $g$, with $t_0$ being the least significant bit (LSB). For each $t_i$, a soft weighted unit clause $(t_i)$ with weight $2^i$ is added.  This yields the weighted MaxSAT formulation:
$$
\left\{
\begin{array}
{lr}
\text{Hard clauses } & \phi  \text{ : a Boolean formula translated from } \varphi \\
\text{Soft clauses }   & 
\left\{
\begin{array}
{lr}
(t_0 \ \  \text{  weight 1}) \  \quad \quad     & \\ 
(t_1 \ \  \text{  weight 2}) \   \quad \quad     & \\ 
\quad \ \ \  \ \ \ldots \ \ \ \ \ \ \ \quad  \quad \quad  &\\
(t_{n - 1} \text{ weight } 2^{n-1}) \quad  \quad   \quad  &
\end{array}
\right.
\end{array}
\right.
$$

In other words, solving the \revise{OMT(BV)} problem boils down to maximizing the objective function $f(t_{n - 1}, \ldots, t_0) = 2^{n - 1}*t_{n - 1} + \cdots + 2^0 * t_0$ subject to $\phi$.

\begin{example}
\label{example:2-bit}
Consider a simple 2-bit unsigned bit-vector $x$ and the formula $\varphi \equiv x > 1$. Intuitively, the maximum value of $x$ subject to $\varphi$ is 3. 
We encode $x$ as two Boolean variables $t_1$ and $t_0$, with $t_1$ representing the most significant bit (MSB) and $t_0$ representing the least significant bit (LSB). 
The formula $\varphi$ is bit-blasted as $t_1 \lor t_0$, and the corresponding MaxSAT instance is as follows:
    $$
\left\{
\begin{array}
{lr}
\text{Hard clauses }   & \text{Boolean formula }  t_1 \lor t_0  \\
\text{Soft clauses }   & 
\begin{array}
{lr}
(t_0 \ \  \text{  weight 1}) \   \quad \quad     & \\ 
(t_1 \ \  \text{  weight 2}) \   \quad \quad     & \\ 
\end{array}
\end{array}
\right.
$$

After maximizing the objective function $f(t_1, t_0) = 2*t_1 + t_0$, we obtain the model $[t_1 = 1, t_0 = 1]$, which corresponds to the maximum value of $x = 3$.
\end{example}

\smallskip
\noindent \textbf{Constraint Solving}.
Algorithm~\ref{alg:basic} illustrates the core approach of~\cite{nadel2016bit}, which maximizes the objective function $f$ through iterative maximization of bits from MSB to LSB.
The basic idea is straightforward: the higher the bit, the greater its weight in the objective function.  
The algorithm proceeds as follows:
\begin{itemize}
     \item \emph{Initialization}: The algorithm starts by initializing an empty set of assumption literals $lits$ (Line~\ref{line:lits}),  which track the bits that can be true. The solver solves the formula $\varphi$ and retrieves a model $M$.
     \item \emph{Iterative maximization}: The algorithm iterates from the most significant bit $t_{n-1}$ down to the least significant bit $t_0$. For each bit $t_i$, it checks if $t_i = 1$ in the model $M$ (Line~\ref{line:check_model_ti}). If true, $t_i$ is added to $lits$, and the process continues to the next bit (Line~\ref{line:check_model_success}). If $t_i = 1$ does not hold in $M$, the algorithm checks if the formula $\phi$ is satisfiable under the assumptions of $lits \cup {t_i}$ (Line~\ref{line:check_next_bit}). If satisfiable, the model $M$ is updated, $t_i$ is set to 1, and $ t_i$ is appended to $lits$. 
     Otherwise, $\neg t_i$ is added to $lits$ (Line~\ref{line:set_ti_zero}).
    \item \emph{Termination}: This process continues until all bits are assigned values. 
    The set $lits$ yields the model that maximizes the value of the objective function, i.e., $2^{n-1}*t_{n-1} + \cdots + 2^0 * t_0$.
\end{itemize}

\begin{example} 
\revise{
Consider a 3-bit objective represented by $t_2,t_1,t_0$ and the Boolean formula $\varphi = (\neg t_2 \lor \neg t_1) \land (t_1 \lor t_0)$. OBV-BS first tests the most significant bit $t_2$. Setting $t_2=1$ is satisfiable, so the algorithm fixes $t_2$. It then tests $t_1$ under this assumption. The conjunction $t_2 \land t_1$ violates the first clause, so $t_1$ is fixed to 0. Finally, $t_0=1$ is required by the second clause and is satisfiable. The resulting maximum is therefore $101_2=5$. The example illustrates why bits within one objective must be decided in significance order: the decision for $t_1$ depends on the earlier decision for $t_2$.}
\end{example}

\subsection{Problem Statement}
The above algorithm can be extended to optimize multiple independent objectives. A straightforward approach is to maximize the objectives sequentially \revise{using the OBV-BS algorithm, Figure~\ref{fig:para_1}}. While effective for cases with few objectives, this approach becomes computationally prohibitive as the number of objectives grows, hindering scalability.

This work aims to scale boxed OMT (BV) solving through parallelization.
We first sketch a simple, \emph{objective-level partitioning} strategy that assigns each parallel worker a distinct objective to optimize independently (\cref{subsec:obj-div}).
We then introduce a finer-grained \emph{bit-level partitioning} algorithm, which partitions the problem at the bit level (\cref{subsec:bit-div}). In this approach, each worker is tasked with optimizing specific bits across multiple objectives, enabling greater parallelism.

%% file: 4divide.tex
\section{Approach}
\label{sec:divide-and-conquer}

This section presents our parallel optimization strategies under the divide-and-conquer paradigm. We first partition the problem at the objective level and subsequently extend our approach to a finer, bit-level parallelization.

\subsection{Objective-Level Divide-and-Conquer}
\label{subsec:obj-div}
We begin with a straightforward approach that partitions the search space by objectives. 
\revise{
The basic idea is that each boxed objective can be optimized as an independent problem and therefore can be assigned to a parallel worker.}

\begin{algorithm}[t]
	\caption{Objective-level, divide-and-conquer parallel solving for abstract transformer synthesis}
	\label{algo:queue}
	\KwIn{A formula $\varphi$, a set of objectives $G = \{ g_1,\ldots, g_n \}$,  and $k$ workers}
	\KwOut{The maximal values of $g_1, \ldots, g_n\ s.t.\ \varphi$}
	
	\SetKwFunction{queue}{solve-by-objs}
	\SetKwProg{Fn}{Function}{}{}
	\Fn{\queue{$\varphi,G$,k}}{
            $q  \gets [\ ]; res \gets [\ ]$\;
            \tcp*{$q$ an empty queue; $res$ stores the result of objectives}
            \ForEach{$g_i \in G$}{
			$enqueue(g_i, i)$ \tcp*{Enqueue each objective}
		}
		\For{$i \in \{1,\dots,k\}$}{
                launch a worker that executes\ $solve(\varphi,q,res)$\;
		}
            wait for all workers to finish\; 
            \Return $res$\;
	}
 
        \SetKwFunction{solve}{solve}
        \SetKwProg{Fn}{Function}{}{}
        \Fn{\solve{$\varphi,q,res$}}{
            \While{q \text{ is not empty}}{
                $g_i,i \gets q.get()$ \tcp*{Dequeue an objective $g_i$ and its index $i$}
                $res[i] \gets \text{OBV\_BS}(\varphi,g_i)$ \tcp*{Solve $g_i$ using the OBV\_BS algorithm}
                \tcp*{store the result of $g_i$ in $res[i]$}
            }
            \Return\;
        }
\end{algorithm}

\smallskip
\noindent \textbf{Basic Procedure}.
Algorithm~\ref{algo:queue} shows how multiple objectives can be optimized in parallel. Given a formula $\varphi$ and a set of objectives $G = \{g_1, g_2, \dots, g_n\}$, the algorithm returns a vector of maximum values for all objectives in~$G$. Each objective is enqueued, and $k$ workers are launched. Each worker repeatedly dequeues one objective at a time, invokes the OBV-BS algorithm to optimize it, and stores the result before moving on to the next objective. 

\revise{The \texttt{solve} function is the core routine executed by parallel workers. Specifically, the main program uses the \textbf{for} loop at Line~7 to spawn $k$ concurrent worker threads. Here, the index $i$ serves solely as a counter for instantiation, meaning each spawned worker simply invokes the \texttt{solve} function independently without needing $i$ as a parameter. Within the \texttt{solve} function, the \textbf{while} loop at Line~12 is responsible for processing the objectives. As long as the shared task queue is not empty, each worker utilizes this loop to dynamically dequeue an objective and compute its optimal value.}
\begin{example}
\label{exmp:idea}
Consider an example with three objectives \( G = \{g_1, g_2, g_3\} \) and two parallel workers \( p_1 \) and \( p_2 \). As illustrated in Figure~\ref{fig:para_2},
\begin{itemize}
    \item At the first step, \( p_1 \) dequeues and solves \( g_1 \), while \( p_2 \) works on \( g_2 \).
    \item After \( p_1 \) finishes solving \( g_1 \), it stores the result \( r_1 \) and proceeds to solve \( g_3 \). Once \( p_2 \) completes \( g_2 \), it terminates, as no more objectives remain in the queue. \( p_1 \) finishes solving \( g_3 \) shortly thereafter.
\end{itemize}
\end{example}

\smallskip
\noindent \textbf{Limitations of the Strategy}.
The above objective-level parallel algorithm exhibits several limitations:
\begin{itemize}
    \item \emph{Imbalanced Workload Distribution}: Variability in objective complexity can lead to uneven computational effort across workers, resulting in idle resources for some and excessive load for others.
   \item \emph{Limited by Sequential Performance}: While the OBV-BS algorithm performs well for many practical instances, its scalability is constrained for complex instances, potentially diminishing the benefits of parallelism.
\end{itemize}

    \begin{example}
Consider again the example from Example~\ref{exmp:idea}, where we have three objectives ($g_1 = \{a_1, a_0\}$, $g_2 = \{b_1, b_0\}$, and $g_3 = \{c_1, c_0\}$), and two parallel workers $p_1$ and $p_2$. 
Figure~\ref{fig:para_2} illustrates the solving status over time.
\begin{itemize}
    \item Workers $p_1$ and $p_2$ can solve $g_1 = \{a_1, a_0\}$ and $g_2 = \{b_1, b_0\}$ in parallel. Assume they finish at approximately the same time.
    \item After finishing $g_1$, worker $p_1$  continues with $g_3$, while $p_2$ becomes idle, as no more objectives are left to solve.
\end{itemize}
\end{example}

The critical reason is that although the objectives \revise{
can be
scheduled as independent optimization problems}, the individual bits within each objective (e.g., $b_1$ and $b_0$) exhibit strict dependencies. Unfortunately, Algorithm~\ref{algo:queue} does not account for these intra-objective dependencies, limiting the effectiveness of parallelization.

\subsection{Bit-level Divide-and-Conquer Solving}
\label{subsec:bit-div}
To address the limitations of objective-level parallelism, we introduce a more granular approach: bit-level divide-and-conquer solving. This method enables finer-grained parallelism by processing individual bits across multiple workers.

\begin{algorithm}[t]
	\caption{Bit-level, divide-and-conquer parallel solving for abstract transformer synthesis}
	\label{algo:bit-level}
	\KwIn{A formula $\varphi$ and a set of objectives $G = \{ g_1,\ldots, g_n \}$, $k$ workers}
	\KwOut{The maximal values of $g_1, \ldots, g_n\ s.t.\ \varphi$}
	
	\SetKwFunction{bit}{bit-level-para}
	\SetKwProg{Fn}{Function}{}{}
	\Fn{\bit{$\varphi,G$,k}}{
            $\textcolor{red}{q} \gets \cup_{g_i\in G} i$  \tcp*{mark the undecided objectives 
            }
            $\textcolor{red}{res} \gets [[\ ]*n ]$  \tcp*{store the result of $n$ objectives 
            }
		\For{$i \in \{1,\dots,k\}$}{
                launch  a worker that executes $solve(\varphi,G,q,res)$\;
		}
            wait\ for\ all\ workers\ finish\;
            \Return $res$\;
	}
 
        \SetKwFunction{solve}{solve}
        \SetKwProg{Fn}{Function}{}{}
        \Fn{\solve{$\varphi,G,q,res$}}{
            $S \gets Solver(),S.add(\varphi)$\;
            $S.check(),M \gets S.get\_model()$\;
            \While{$q$}{
                $i \gets q.get(),\textcolor{red}{r \gets res[i]}$  \tcp*{get the latest result of $g_i$}
                \label{line:begin}
                $lit \gets $ the\ next\ undecided\ bit\ of\ $g_i$\;
                \If{$r \cup lit\ in\ M$}{
                    $r.append(lit)$  \tcp*{update the local result}
                    \label{line:r.append(lit)}
                }\Else{
                    \If{$S.check(r\cup lit)\ ==\ SAT$}{
                    \label{line:s.check}
                        $r.append(lit)$  \tcp*{update the local result}
                        \label{line:check_append(lit)}
                        $M \gets S.get\_model()$  \tcp*{update the model $M$}
                    }\Else{
                        $r.append(\neg lit)$  \tcp*{update the local result}
                        \label{line:check_append(neglit)}
                    }
                }
                $\textcolor{red}{res[i]\gets r}$  \tcp*{update the result of $g_i$}
                \If{$lit\ is\ not\ the\ last\ bit\ of\ g_i$}{  
                \tcp*{$g_i$ has not been decided}
                    $enqueue(i)$\;
                    \label{line:enqueue_i}
                }
            }
            \Return\;
        }
\end{algorithm}

\begin{figure}[t]
	\centering 
	\begin{subfigure}{0.48\textwidth}
		\includegraphics[width=\linewidth]{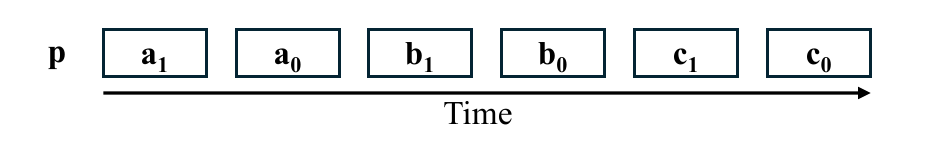}	
		\caption{Sequential solving (Algorithm~\ref{alg:basic} on each objective)}
		\label{fig:para_1}
	\end{subfigure}
	\begin{subfigure}{0.48\textwidth}
		\includegraphics[width=\linewidth]{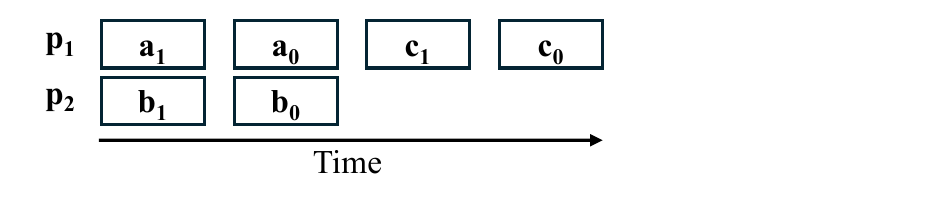}
		\caption{Objective-level parallel solving (Algorithm~\ref{algo:queue})}
		\label{fig:para_2}
	\end{subfigure}
        \begin{subfigure}{0.48\textwidth}
		\includegraphics[width=\linewidth]{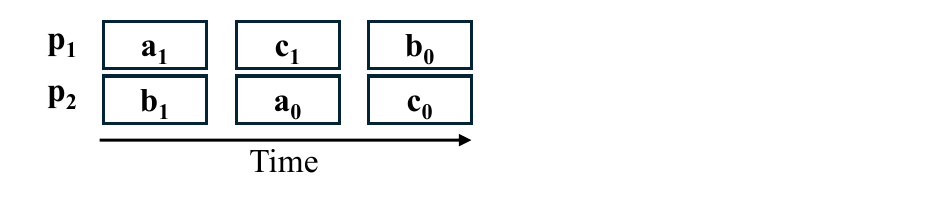}
		\caption{Bit-level parallel solving (Algorithm~\ref{algo:bit-level})}	
		\label{fig:para_3}
	\end{subfigure}
	\caption{Comparing the algorithms for $g_1  = \{a_1, a_0\}, g_2 = \{b_1, b_0\}, g_3 = \{c_1, c_0\}$}
	\label{fig:para}
\end{figure}

\smallskip
\noindent \textbf{Exploiting Bit-Level Independence}.
Our core idea is to leverage the independence of bits across objectives. \revise{
While intra-objective bits (e.g., $b_1$ and $b_0$) have inherent dependencies because OBV-BS must decide them from MSB to LSB, boxed OMT imposes no ordering between the bit decisions of different objectives.} This independence enables workers to process bits concurrently, significantly increasing parallelism.

\begin{example}
Consider the example from the previous section (Example~\ref{exmp:idea}). In this case, the bits $a_1$ and $a_0$ within the same objective are dependent because the value of $a_1$ must be determined before $a_0$. \revise{
However, bits from different objectives can be scheduled concurrently.} For example, in 
Figure~\ref{fig:para_2} and Figure~\ref{fig:para_3}, the bits $a_1$ and $c_1$ can be processed simultaneously without conflict. \revise{
Similarly, in Figure~\ref{fig:para_3}, the decisions for $a_0$ and $b_1$ can proceed in parallel because they belong to different objectives.}
\end{example}

Unfortunately, the objective-level partitioning strategy cannot leverage the independence between \revise{\sout{$a_0$ and $a_1$} bits from different objectives}, thereby limiting the achievable parallelism.

\smallskip
\noindent \textbf{Bit-Level Parallel Search Algorithm}.
Based on the above idea, we introduce the bit-level divide-and-conquer strategy (Algorithm \ref{algo:bit-level}), which processes each objective at the bit level, enabling concurrent processing of bits from different objectives. 
The algorithm has the following key components.

\smallskip 
\noindent \emph{Initialization}. The algorithm initializes a queue of undecided objectives, $q$, and a shared result list, $res$, to store the partial solutions for each objective. 
The algorithm then spawns $k$ parallel workers. Each worker is assigned an undecided objective, processes it one bit at a time, and shares their progress with other workers through the shared result list.

\smallskip 
\noindent \emph{Bit Decision Process}.  Each parallel worker performs the following steps:
\begin{itemize}
    \item Selects an undecided objective $g_i$ from the queue; Retrieves the current partial solution (Line~\ref{line:begin}).
    \item Processes the next undecided bit, represented as $lit$: If the bit $lit$ is already included in the model, it is appended to the local result (Line~\ref{line:r.append(lit)}). Otherwise, the worker checks if appending $lit$ to the partial solution satisfies the overall constraints. If it does, $lit$ is appended to the result (Line ~\ref{line:check_append(lit)}); if not, its negation, $\neg lit$, is appended instead (Line~\ref{line:check_append(neglit)}).
\end{itemize}

\begin{figure*}[t]
	\centering 
	\includegraphics[width=0.78\textwidth]{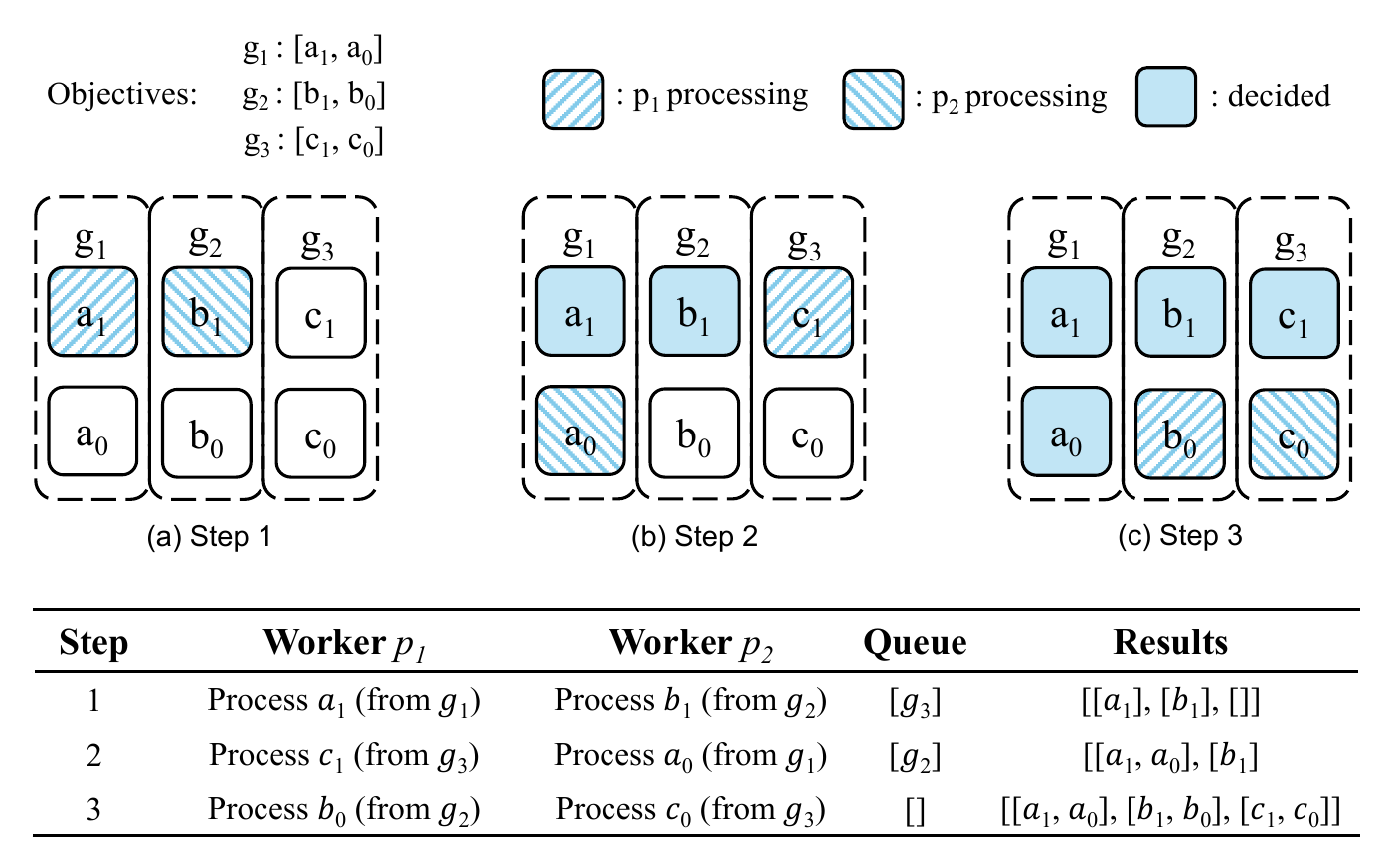}	
	\caption{An example of running Algorithm~\ref{algo:bit-level}}
	\label{fig:bits-par}
\end{figure*}

\begin{example}
Consider the instance again in Example~\ref{exmp:idea}, which has three objectives, each consisting of two bits: $g_1 = \{a_1, a_0\}, g_2 = \{b_1, b_0\},  g_3 = \{c_1, c_0\}$.
This example assumes two available workers, \( p_1 \) and \( p_2 \).
In addition to Figure~\ref{fig:para_3}, we show more details about running the algorithm in Figure~\ref{fig:bits-par}.

\begin{itemize}
   \item In \text{Step 1}, the algorithm initializes with an empty results list and a queue containing all three objectives: \( [g_1, g_2, g_3] \). Worker \( p_1 \) dequeues \( g_1 \) and processes its first bit, \( a_1 \). The worker checks whether assigning \( a_1 = 1 \) satisfies the formula \( \varphi \). If so, \( a_1 = 1 \) is appended to the result for \( g_1 \); otherwise, \( a_1 = 0 \) is appended. Concurrently, worker \( p_2 \) dequeues \( g_2 \) and processes its first bit, \( b_1 \), in a similar manner. After this step, the queue contains only \( g_3 \), and the results list is updated to \( [[a_1], [b_1], []] \).
   
   \item In \text{Step 2}, worker \( p_1 \) processes the first bit of \( g_3 \), which is \( c_1 \). Simultaneously, worker \( p_2 \) returns to \( g_1 \) and processes its second bit, \( a_0 \). \revise{
   After both workers complete their tasks, \( g_1 \) is finished, while both \( g_2 \) and \( g_3 \) still have undecided low bits. The queue therefore contains \( [g_2,g_3] \). The results list is updated to \( [[a_1, a_0], [b_1], [c_1]] \).} 
   
   \item In \text{Step 3}, worker \( p_1 \) processes the second bit of \( g_2 \), which is \( b_0 \), while worker \( p_2 \) processes the remaining bit of \( g_3 \), which is \( c_0 \). At this point, the queue is empty, as all objectives have been fully processed. The final results are \( [[a_1, a_0], [b_1, b_0], [c_1, c_0]] \). 

\end{itemize}
\end{example}

The above strategy overcomes many limitations of objective-level parallelism by decomposing objectives into individual bits. It exploits inherent bit-level independence across different objectives, enabling greater concurrency and more efficient utilization of computational resources.

\smallskip
\noindent \textbf{Scheduling Considerations}.
Bit-level parallelism addresses several shortcomings of objective-level parallelism but also faces several challenges.
The fine-grained nature may incur overhead due to frequent task scheduling and synchronization: Each worker must regularly poll the task queue, retrieve bits, and update the shared result array, which can lead to inefficiencies.
To mitigate excessive synchronization, we extend Algorithm~\ref{algo:bit-level} to optionally group multiple bits (up to a user-defined parameter $n$) into a single task. This reduces the frequency of queue operations at the cost of somewhat coarser parallelism.
We will evaluate the impact of choosing different $n$ in the following section.

\section{\revise{Implications and Applications}}
\label{sec:applications}

We have implemented \ToolName\ using Z3~\cite{de2008z3} and PySAT~\cite{ignatiev2018pysat}. Z3 is responsible for bit-blasting the bit-vector formulas and objectives into Boolean constraints and variables, while PySAT provides the SAT solving infrastructure.

\revise{
The objectives are derived from the client's analysis as the numeric quantities the client wants to maximize or minimize. Concretely, each OMT query is produced from a client-generated path condition, and the boxed objectives correspond to the bounds that are semantically relevant to that client. }

\smallskip
\noindent \revise{\textbf{Constrained Random Fuzzing.} We implement Pangolin’s template-polyhedron abstraction of path conditions~\cite{huang2020pangolin} on top of Angr~\cite{shoshitaishvili2016sok}. The OMT instances arise from symbolic execution along a selected set of program paths. For each path prefix targeted by the fuzzer, Angr constructs a bit-vector path condition over the symbolic state and instantiates a fixed polyhedral template over the input variables. It then treats the template coefficients (equivalently, the facet bounds) as optimization objectives. Optimizing these objectives yields a family of constraints that delimit the admissible region of the relevant inputs, thereby concentrating mutations on values that remain consistent with the chosen prefix. The resulting boxed OMT query is obtained directly from the path condition, with one objective per template coefficient, and its solution parameterizes the template polyhedron that generates high-coverage seeds.}

\smallskip
\noindent \revise{\textbf{Static Binary Bug Detection.} In Clearblue~\cite{ye2024manta,zhou2024plankton}, OMT queries are derived from the block-level transfer functions of the static analysis. For each selected function or basic block, the analyzer constructs a bit-vector formula that captures the precise semantics of register and memory updates (including aliasing and arithmetic at machine width), and then optimizes for the tightest interval bounds on the state components of interest. The objectives are the lower and upper bounds of these registers and memory values, since these bounds directly control whether the analysis can discharge safety checks, such as out-of-bounds accesses or arithmetic overflow. \ToolName\ solves the resulting boxed objectives and returns a sharper symbolic interval abstraction, which the client uses to strengthen invariants and reduce false positives.}

%% file: 5evaluation.tex
\section{Evaluation}
\label{sec:eval}

Our evaluation is guided by the following research questions:
\begin{itemize}
    \item \textbf{RQ1:} How does \ToolName\ compare to state-of-the-art OMT(BV) solvers in terms of solved instances, runtime, and scalability? (\cref{subsec:eval:other})
    \item \textbf{RQ2} How do different partitioning granularities affect the performance of \ToolName? (\cref{subsec:eval:ablation})
\end{itemize}

\noindent \textbf{Benchmarks}. 
\revise{
Table~\ref{tbl:benchmarks} provides relevant information about the benchmarks, including the source programs from which the OMT queries are generated, the number of queries, the average number of objectives, the range of objective bit-widths, and the range of clause counts. These queries are generated using two analyzers:} (1) Angr \cite{shoshitaishvili2016sok}, a binary analysis framework with a symbolic execution engine for hybrid fuzzing. We implement a template polyhedron abstraction of path conditions~\cite{huang2020pangolin}, thereby improving seed mutations. (2) Clearblue~\cite{ye2024manta,zhou2024plankton}, a static binary analyzer that translates binaries into LLVM IR for vulnerability detection.
We use symbolic abstraction to compute symbolic interval abstractions of register and memory values to detect vulnerabilities. We exclude queries that can be solved by sequential Spear within 5 seconds.

\revise{Note that all extracted queries are formulated as independent multi-objective optimization problems, which is the primary focus of \ToolName. We therefore exclude traditional single-objective OMT benchmarks (e.g., those from Nadel et al.~\cite{nadel2016bit}), since our parallelization strategies are tailored to multi-objective settings.}


\smallskip
\noindent \textbf{OMT(BV) Solvers}.
We evaluate \ToolName\ against two state-of-the-art and widely-used solvers for boxed OMT(BV):
(1) $\nu$Z~\cite{bjorner2015nuz}: An extension of Z3 supporting optimization problems, which reduces OMT(BV) to weighted MaxSAT solving. We employ a parallelized configuration of $\nu$Z, leveraging Z3's parallel SAT-solving capabilities to distribute SAT queries across multiple cores;
(2) OptiMathSAT~\cite{sebastiani2020optimathsat}: An extension of MathSAT with OMT support. OptiMathSAT includes multiple engines for solving boxed OMT problems. We configure a portfolio-based version by diversifying its solver configurations.
Our objective-level parallel algorithm often falls short of the bit-level one. Hence, this section mainly reports the results of the bit-level algorithm.

\smallskip
\noindent \textbf{Environment}.
All experiments are conducted on a server running Ubuntu 22.04 LTS, equipped with an Intel Xeon E5-2640 v4 CPU (2.40 GHz) and 500 GB of RAM. Each solver is executed once per benchmark instance, with a timeout threshold of 120 seconds. We choose the timeout based on two key observations: (1) the vast majority of queries are resolved within this time frame, and (2) the target application is program analysis, where the analysis of a single program can generate a large number of queries. 

\subsection{Comparison with Existing OMT Solvers}
\label{subsec:eval:other}

\begin{table*}[t]
    \centering
    \caption{\revise{Benchmark information}}
    \resizebox{\textwidth}{!} 
    {
    \begin{tabular}{lccccccc}
    \toprule
    \textbf{Program}  & \textbf{Source} & \textbf{Description} & \textbf{Size (KLoC)}  & \textbf{OMT queries} &\textbf{Avg\_obj} &\textbf{Obj\_bw} &\textbf{Clauses(k)}\\ \midrule
    mcf  & SPEC & Vehicle scheduling& 2 & 211 &16.7 &8-1592 &400-800\\
    crafty  & SPEC & Processor benchmark & 13 & 2 &74.5 &1-64 &414\\
    eon     & SPEC & 3D animation rendering  & 22 & 8 &70.6 &1-64 &330-1000+\\
    gap     & SPEC & Group theory calculations & 36  & 71 &69.1 &1-64 &256-1700+\\
    vortex & SPEC & Object-oriented DBMS & 49  & 7 &44.6 &1-64 &750-1500+\\
    perlbmk & SPEC & A subset of Perl & 73  & 85 &55.1 &1-64 &57-1000+\\ 
    gcc     & SPEC & C optimizing compiler & 135 & 11 &76.6 &1-64 &146-149\\ \hline
    tmux    & OSS  & Terminal multiplexer & 40  & 52 &26.9 &1-2560 &28-128\\
    
    libssh & OSS & C library for ssh & 44 &3 &60 &1-64 &201-4000+\\
    transmission & OSS & BitTorrent client& 88 &5 &44 &1-64 &882-1500+\\
    openssl & OSS & Encryption/decryption lib  & 300 &18 &53.7 &1-64 &17-8000+\\
    wrk     & OSS  & HTTP benchmarking tool & 340 &77 &53.8 &1-64 &148-2700+\\
    glusterfs & OSS & Gluster volume manager& 481 &100 &66.4 &1-64 &108-1000+\\
    php    & OSS &  A scripting language & 863  &322 &26.5 &8-1216 &36-59\\
    wget    & OSS & Software package for retrieving files & 50 & 242 &71.5 &8-8528 &28-102\\
    \bottomrule
    \end{tabular}
   }
    \label{tbl:benchmarks}
\end{table*}

\begin{figure*}[t]
    \centering
    \begin{subfigure}{0.48\textwidth}
        \centering
        \includegraphics[width=0.7\linewidth]{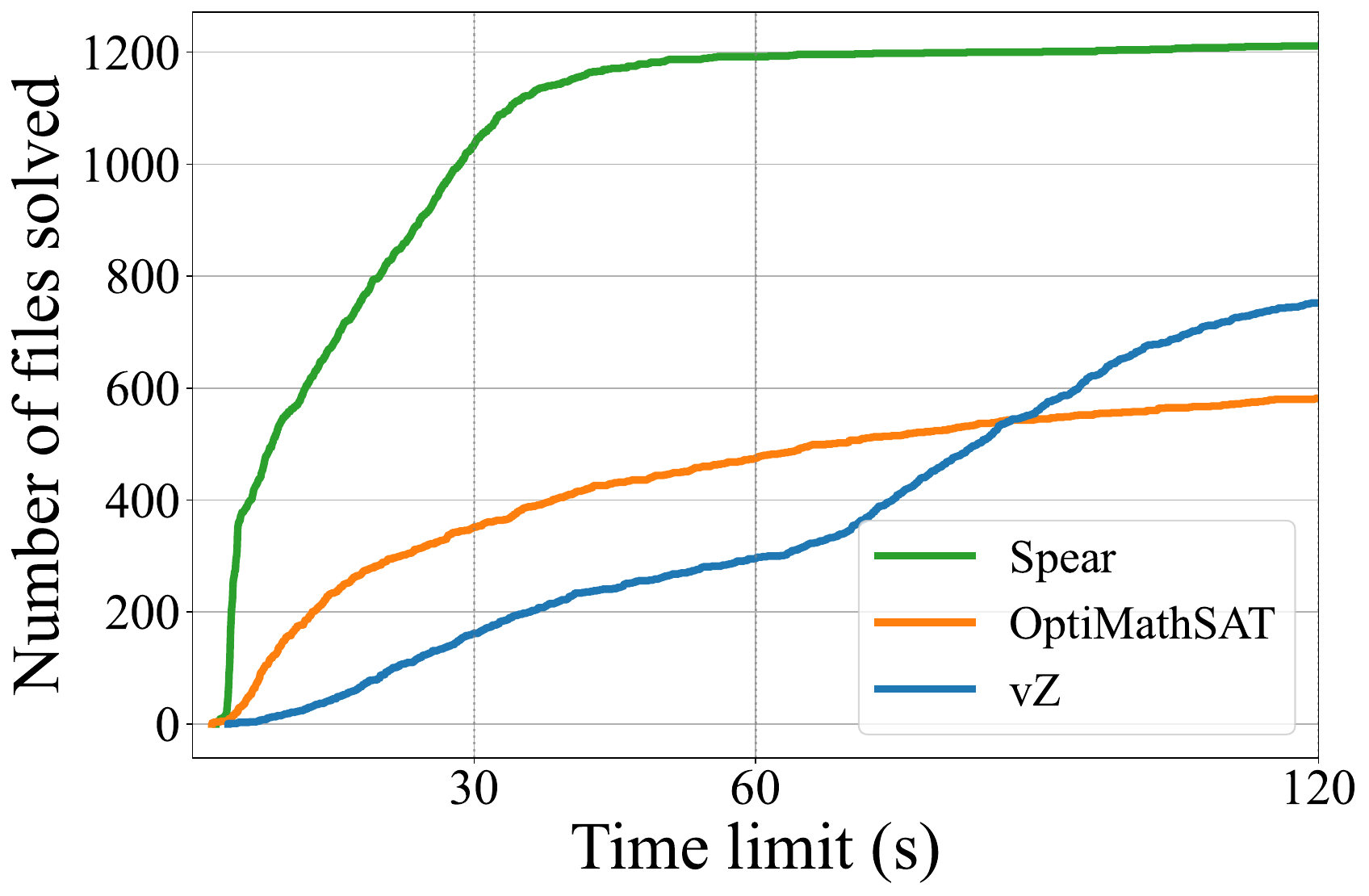}
        \caption{4 cores}
        \label{fig:solvednum-4}
    \end{subfigure}
    \hfill
    \begin{subfigure}{0.48\textwidth}
        \centering
        \includegraphics[width=0.7\linewidth]{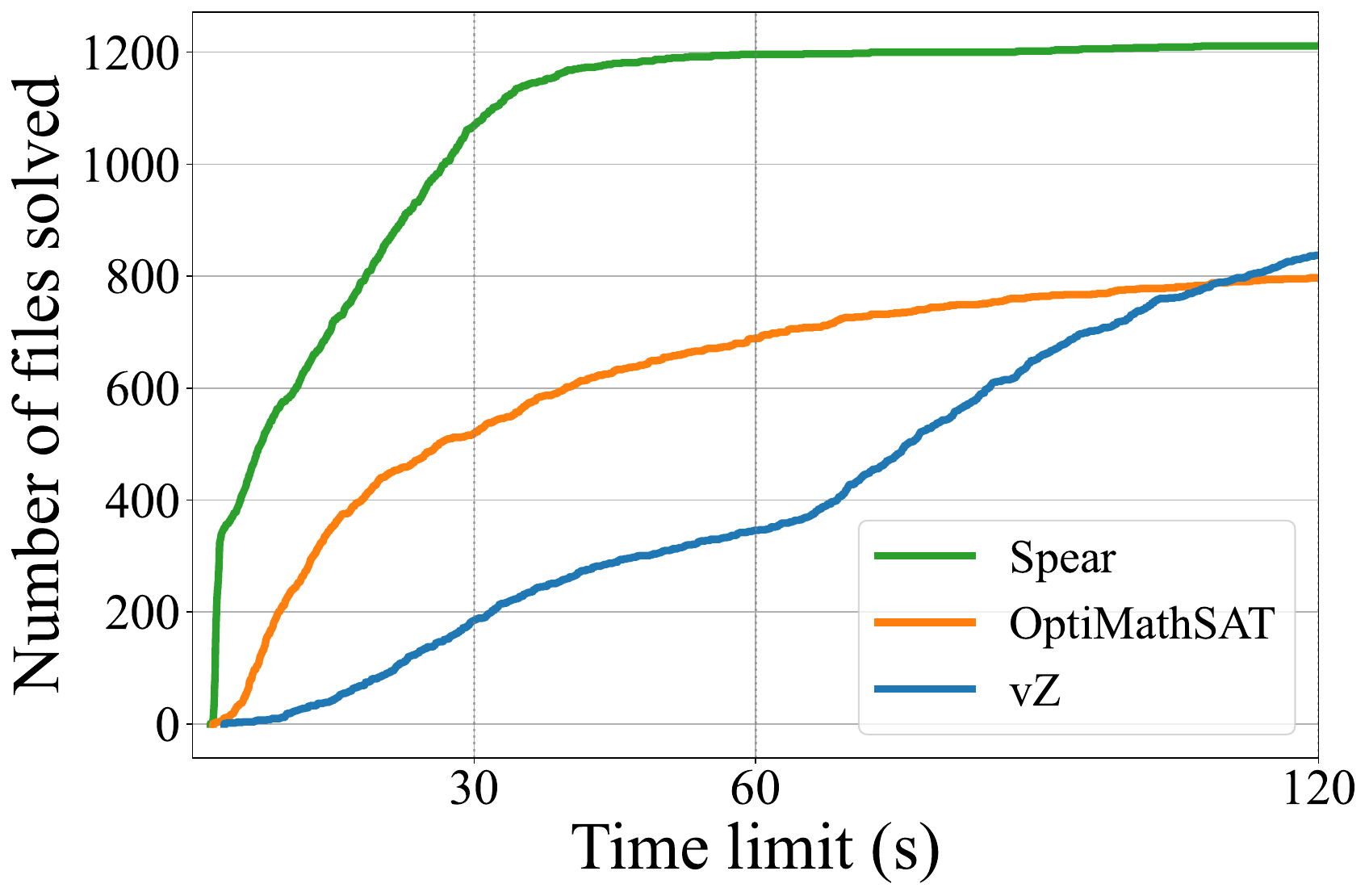}
        \caption{8 cores}
        \label{fig:solvednum-8}
    \end{subfigure} 
    \caption{\revise{Solved Files Comparing \ToolName\ vs. $\nu$Z and OptiMathSAT}}
    \label{fig:solved-num}
\end{figure*}

\begin{figure*}[t]
    \centering
    \begin{subfigure}{0.48\textwidth}
        \centering
        \includegraphics[width=0.7\linewidth]{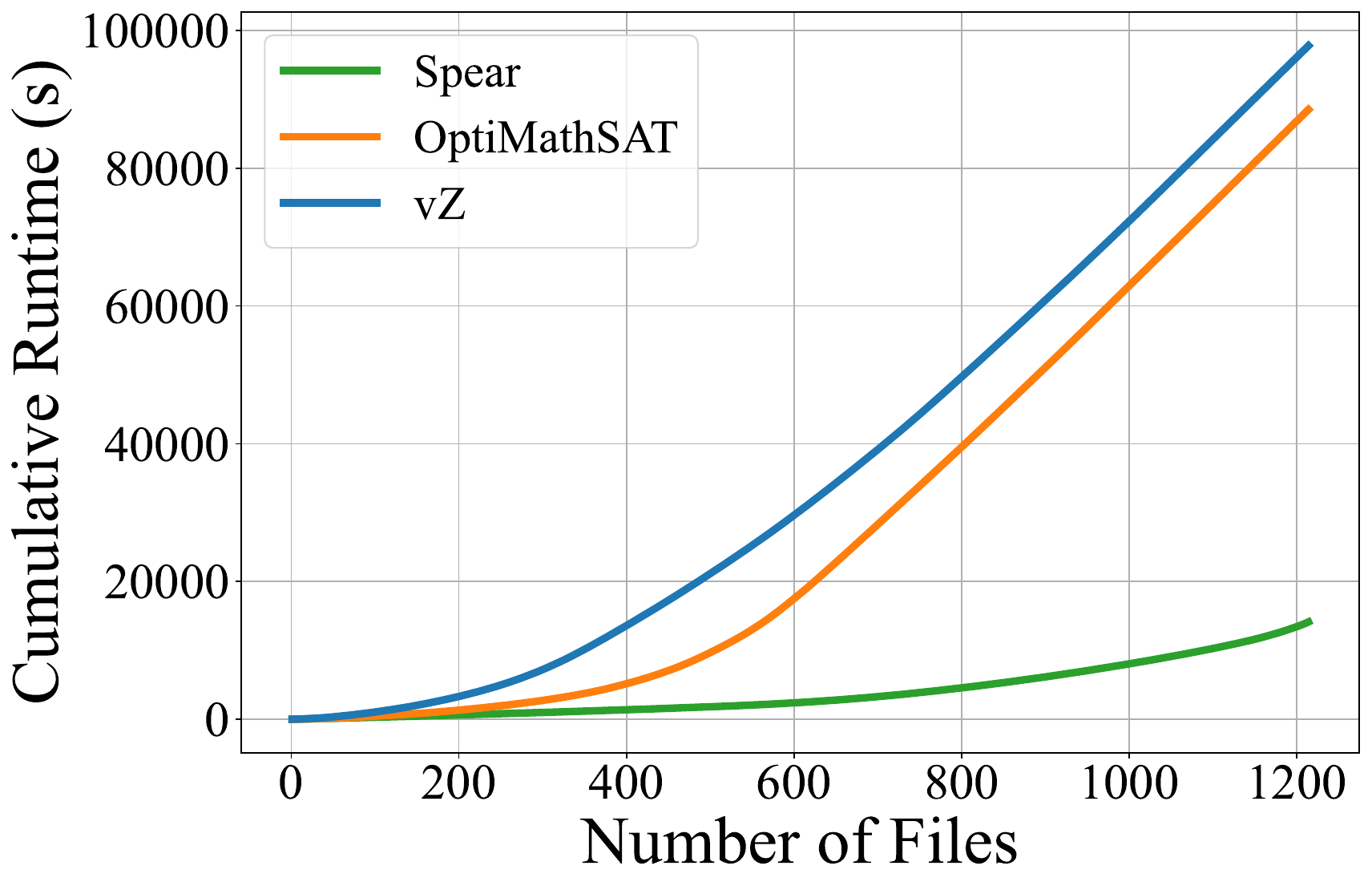}
        \caption{4 cores}
        \label{fig:external-1}
    \end{subfigure}
    \hfill
    \begin{subfigure}{0.48\textwidth}
        \centering
        \includegraphics[width=0.7\linewidth]{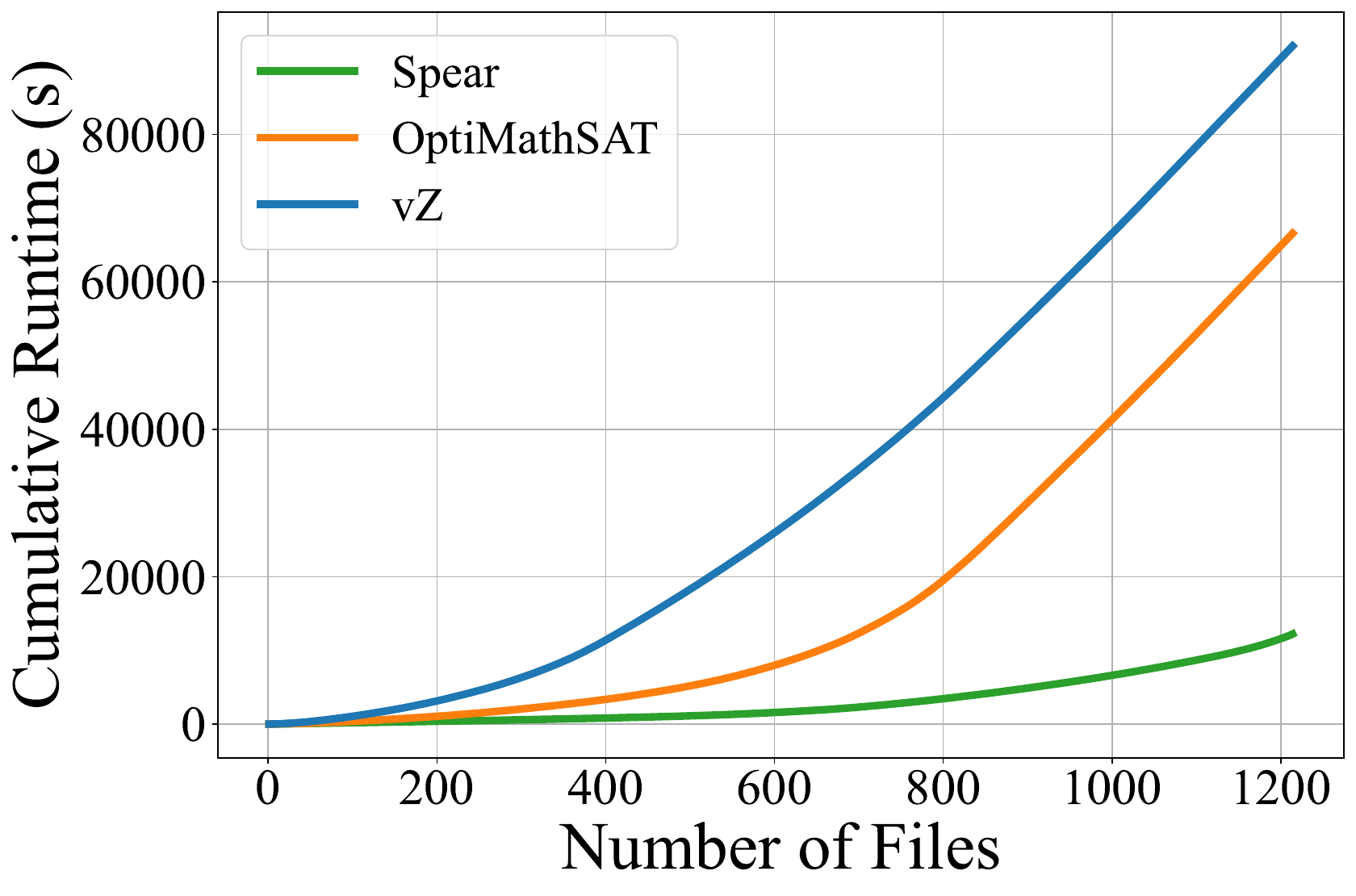}
        \caption{8 cores}
        \label{fig:external-2}
    \end{subfigure} 
    \caption{Comparing \ToolName\ vs. $\nu$Z and OptiMathSAT}
    \label{fig:solved-internal}
\end{figure*}

\begin{figure}[t]
    \centering
    \begin{subfigure}{0.49\linewidth}
        \centering
        \includegraphics[width=\linewidth]{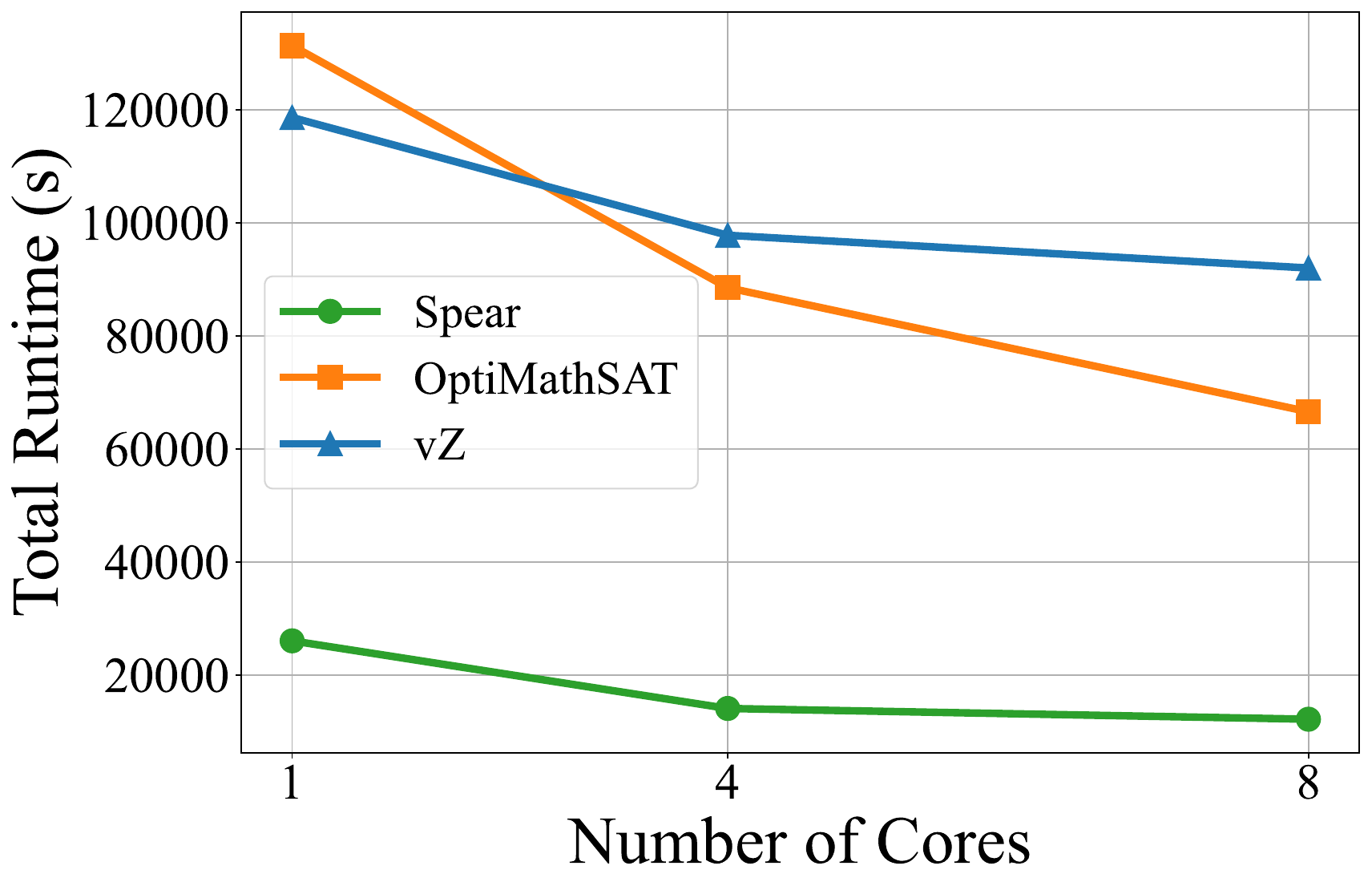}
        \caption{Total runtime}
        \label{fig:scal-total}
    \end{subfigure}
    \hfill
    \begin{subfigure}{0.47\linewidth}
        \centering
        \includegraphics[width=\linewidth]{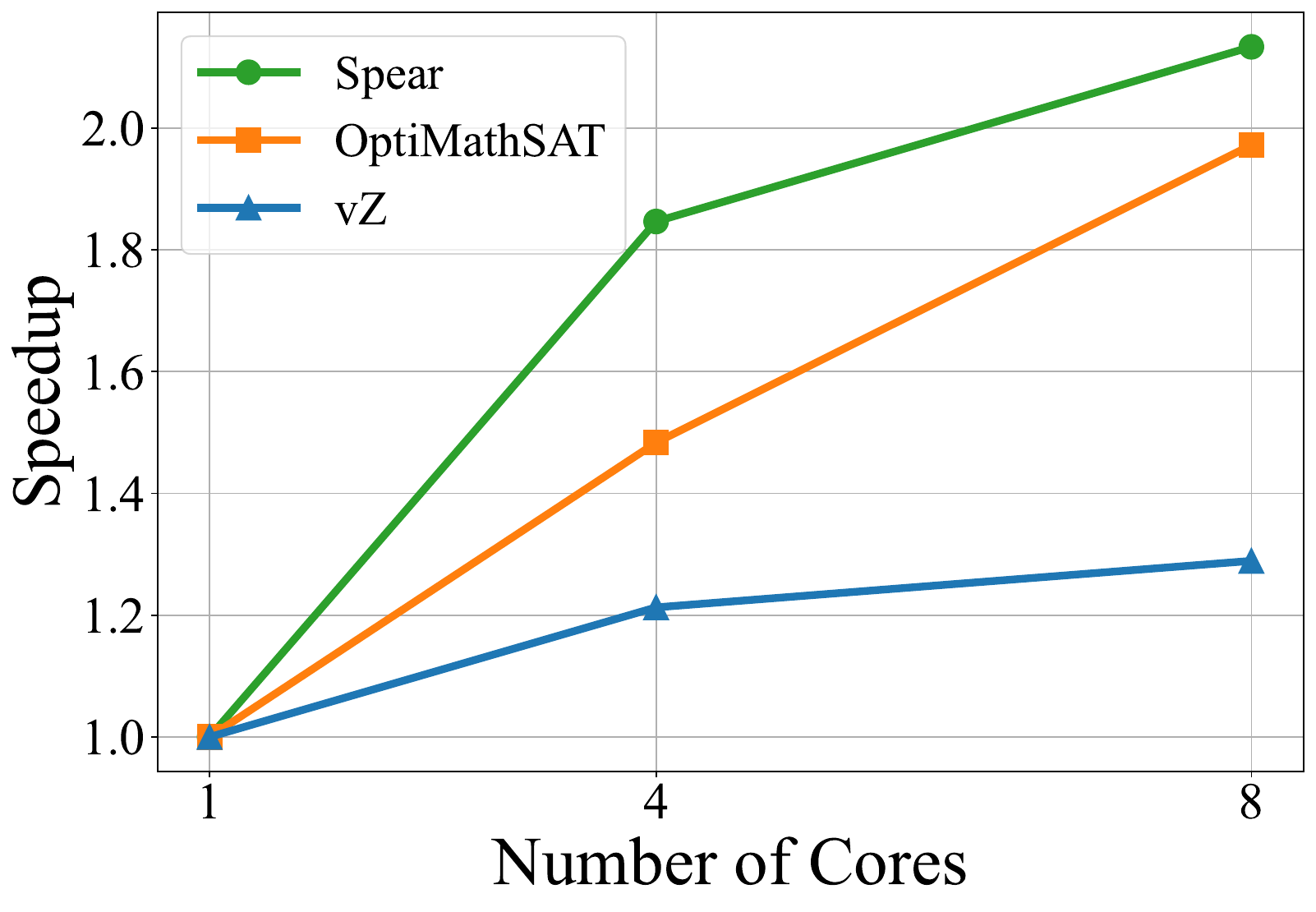}
        \caption{Speedup ($T_1/T_p$)}
        \label{fig:scal-speedup}
    \end{subfigure}
    \caption{\revise{Scalability of \ToolName, $\nu$Z, and OptiMathSAT as the number of cores grows from 1 to 8. Speedup is computed per solver relative to its own single-core total runtime $T_1$.}}
    \label{fig:scalability}
    \vspace{-3mm}
\end{figure}

\begin{figure*}[t]
    \centering 
    \includegraphics[width=\textwidth]{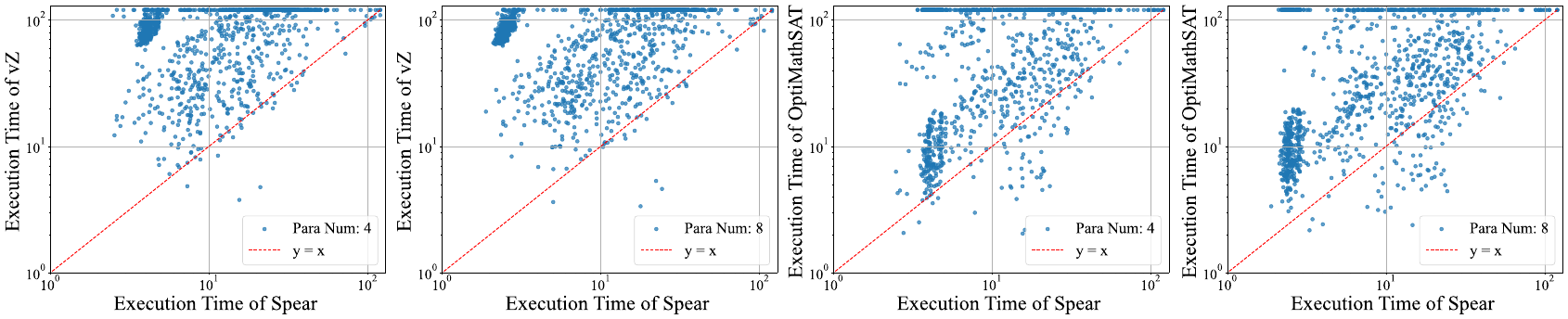}	
    \caption{Scatter plots comparing \ToolName\ against $\nu$Z and OptiMathSAT.}
    \label{fig:external}
\end{figure*}

We compare the best configuration of \ToolName\ (grouping 8 bits in a single task) against $\nu$Z and OptiMathSAT under varying configurations.
Figures~\ref{fig:solved-num}-\ref{fig:external} summarize the results.

\smallskip
\noindent \textbf{Solved Instances}.
\ToolName\ solves substantially more instances than both competitors across all core configurations. 
With 8 cores, \ToolName\ solves 1,211 instances—45\% more than 
$\nu$Z (838) and 52\% more than OptiMathSAT (797). 
\revise{Beyond the final counts, the relationship between the number of
solved instances and the time limit further highlights the advantage of
\ToolName\ (8 cores), as shown in Figure~\ref{fig:solved-num}. At a 30s limit, \ToolName\ already solves 1,069 instances, far more than OptiMathSAT (520) and $\nu$Z (185); at 60s the counts grow to 1196, 689, and 346, and at 120s to 1211, 797, and 838,
respectively. \ToolName\ is strongly front-loaded: it solves 88\% of its instances within 30s and 99\% within 60s, so further enlarging the time limit barely changes its solved count.
In contrast, the competitors keep solving more instances as the budget grows, but only at the price of substantially longer solving times, since the additional instances are precisely those that require close to the
full timeout. \ToolName\ thus dominates both solvers under every time
limit, and its margin is largest at the tightest budget—exactly the regime
most relevant to the query-intensive program-analysis setting that
motivates this work.}

\smallskip
\noindent \textbf{Runtime Performance}.
 \ToolName\ also achieves significantly lower solving times. With eight cores, its average time per instance (10.1s) is $7.5\times$ faster than 
$\nu$Z (75.8s) and $5.4\times$ faster than OptiMathSAT (54.8s). The total runtime (12,225.2s) is also $7.5\times$ smaller than 
$\nu$Z's despite solving 373 more instances. Furthermore, \ToolName\ exhibits strong scalability: increasing cores from 1 to 8 reduces average time by 53\%, compared to 22\% for $\nu$Z and 49\% for OptiMathSAT. Overall, the result highlights the effectiveness of \ToolName's fine-grained bit-level parallelism.

\revise{To make this scalability behavior explicit—rather than leaving it implicit in the per-file cumulative plots — Figure~\ref{fig:scalability} reports, for each solver, how the total solving time and the resulting speedup change as the number of cores grows from 1 to 8. Two observations stand out. First, in absolute terms (Figure~\ref{fig:scal-total}) \ToolName\ is far below both competitors at every core count: its total time drops from 26{,}081s on a single core to 14{,}124s (4 cores) and 12{,}225s (8 cores), roughly $5$--$7.5\times$ smaller than $\nu$Z and OptiMathSAT throughout. Second, in relative terms (Figure~\ref{fig:scal-speedup}) \ToolName\ also scales best, reaching a $2.13\times$ speedup at 8 cores, versus $1.97\times$ for OptiMathSAT and only $1.29\times$ for $\nu$Z. \ToolName's curve is steepest between 1 and 4 cores ($1.85\times$) and flattens afterward, indicating that four cores already exploit most of the available bit-level parallelism, with diminishing returns beyond that point. In contrast, $\nu$Z barely benefits from additional cores, as its weighted-MaxSAT reduction offers limited independent work to distribute. We emphasize that speedup is normalized per solver against its own single-core runtime, so it measures parallel efficiency rather than absolute performance; the two views are therefore complementary, and \ToolName\ leads on both.}

\ToolName\ solves nearly all instances within 20 seconds, while competitors require progressively longer times.
As illustrated in the scatter plots \revise{(where each dot represents a single OMT query, and its coordinates correspond to the solving time in seconds for the two compared setups)}, \ToolName\ outperforms both solvers in the vast majority of cases. \revise{Specifically, dots located above the diagonal line indicate instances where \ToolName\ (plotted on the x-axis) is faster than the competitor (plotted on the y-axis).} For example, against  $\nu$Z and OptiMathSAT with eight cores, \ToolName\ is faster on 97\% and 94\% of commonly solved instances\revise{\sout{ (points below the diagonal)}}, respectively. Only 6\% of instances favor OptiMathSAT, likely because its diversified, heuristic-driven MaxSAT approach outperforms bit-level search on specific constraints.

\smallskip
\noindent \textbf{Diminishing Returns}.
\ToolName\ achieves strong scalability for four cores, effectively solving most tractable problems within this parallelism regime. Only a small subset of intrinsically difficult cases remains, exhibiting limited gains from additional parallelization.
This performance gap may indicate bottlenecks localized to specific critical objectives or bits. A promising optimization direction could involve integrating a portfolio method to address these residual hard instances through complementary algorithmic strategies.



\smallskip \noindent \textbf{Practical Impact}. Notably, a 2–3$\times$ speedup often distinguishes an infeasible analysis pipeline from one suitable for production environments. For instance, analyzing a large codebase such as MySQL (2 MLoC) initially took over 10 hours. Reducing analysis time by several hours can make bug detection and variant testing viable within typical industrial time budgets of 5–10 hours per analysis~\cite{mcpeak2013scalable}.

\subsection{Ablation Study}
\label{subsec:eval:ablation}

\begin{table}[t]
\centering
\caption{Comparing variants of \ToolName\  with different partitioning bits}
\label{tbl:internal}
\begin{tabular}{l  c c}\toprule
\textbf{Solver}  & \textbf{Time (s)} & \textbf{Avg. Time (s)}  \\ \midrule
    \ToolName-1 bit (4 cores) &  22,561.6& 18.6 \\
   \ToolName-4 bits (4 cores)  & 15,397.0      & 12.7     \\
   \ToolName-8 bits (4 cores)    & 14,124.4      & 11.6   \\ \hline
   \ToolName-1 bit (8 cores)    & 18,084.9      & 14.9  \\ 
   \ToolName-4 bits (8 cores)   & 12,901.6      & 10.6  \\
    \ToolName-8 bits (8 cores)   & 12,225.2      & 10.1   \\ 
  \bottomrule
\end{tabular}
\end{table}

\begin{figure*}[t]
	\centering 
	\begin{subfigure}{0.49\textwidth}
			\centering
		\includegraphics[width=0.7\linewidth]{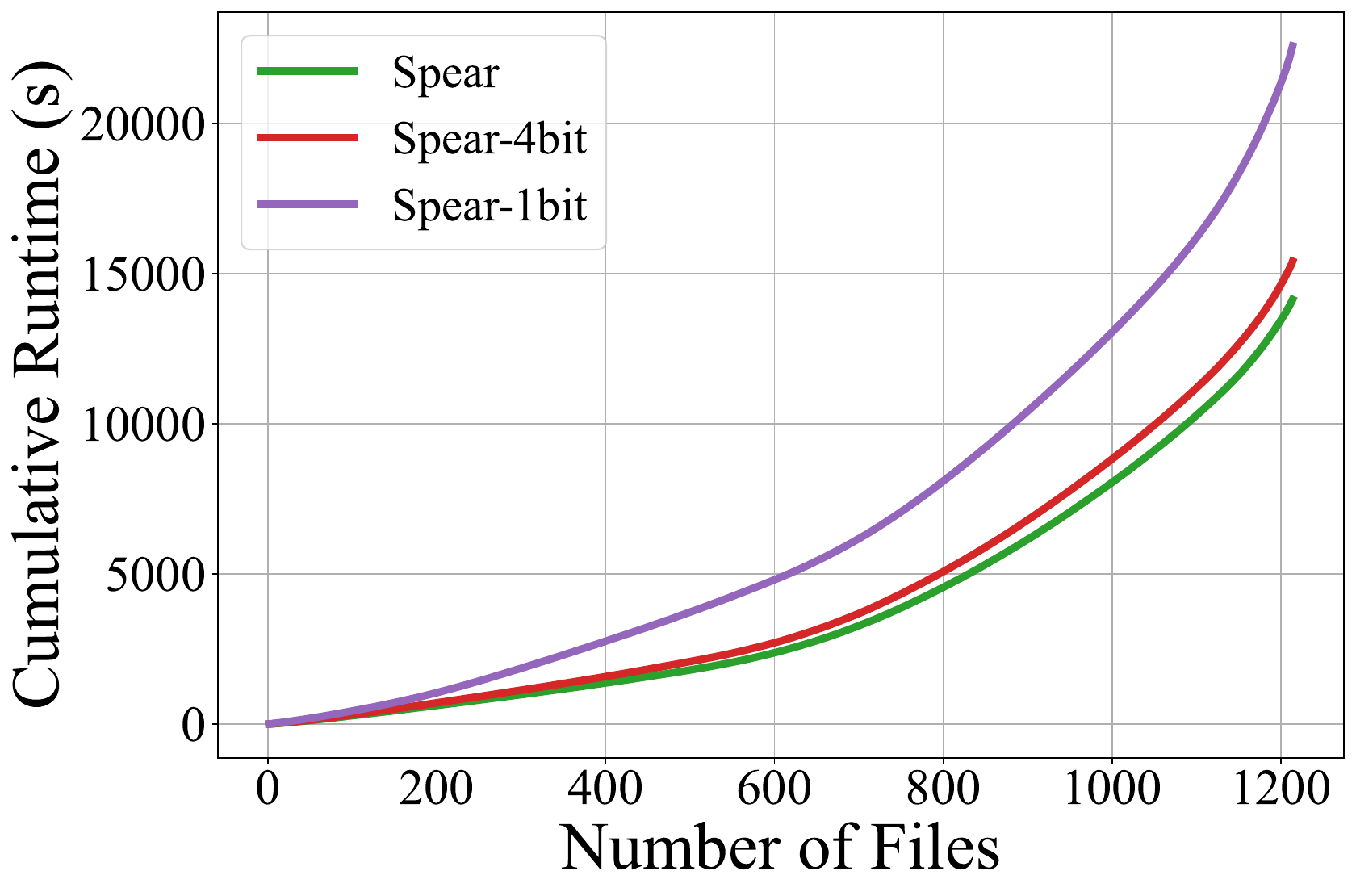}	
		\caption{4 cores}
		\label{fig:internal-1}
	\end{subfigure}
	\begin{subfigure}{0.49\textwidth}
			\centering
		\includegraphics[width=0.7\linewidth]{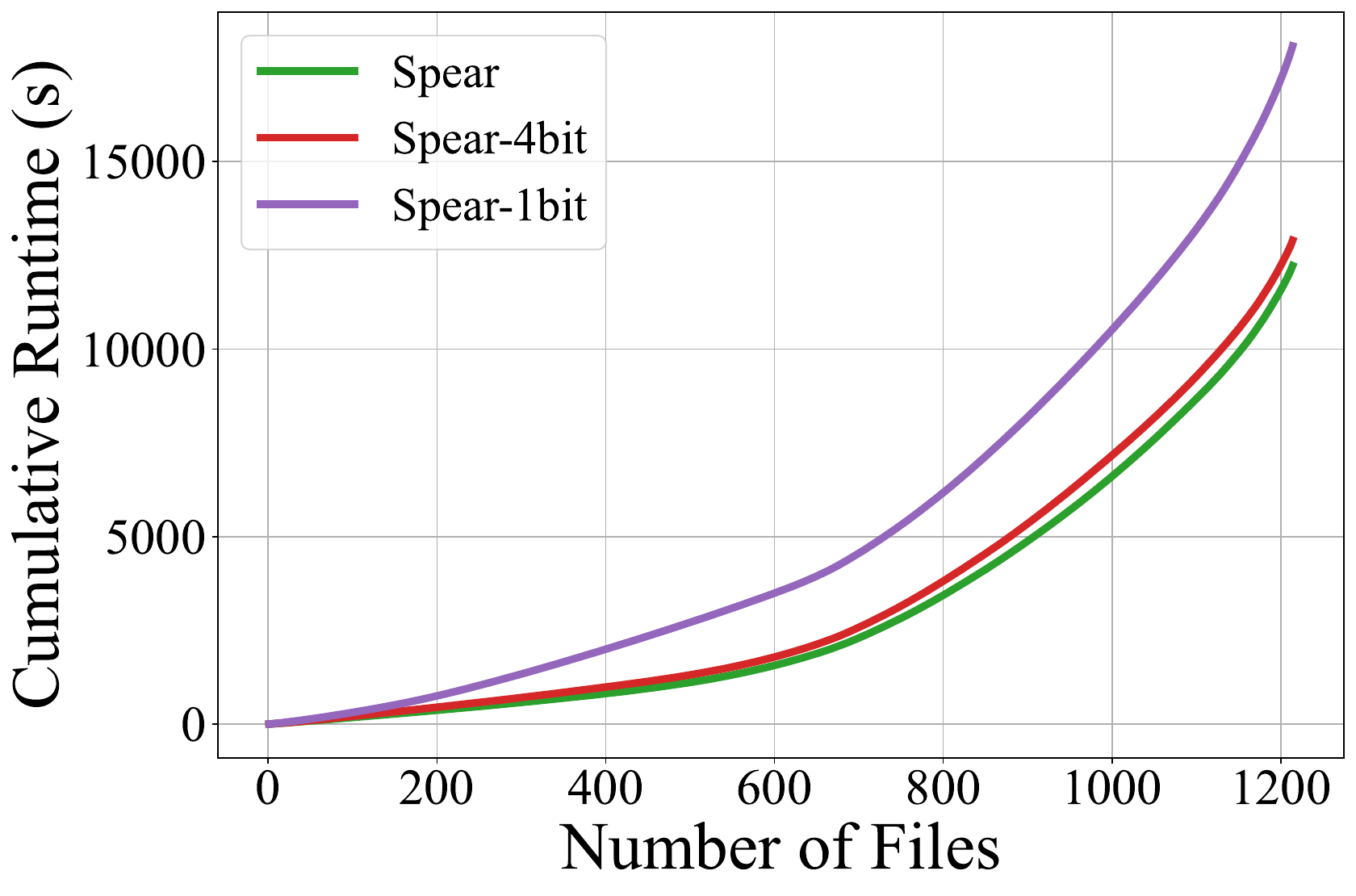}
		\caption{8 cores}
		\label{fig:internal-2}
	\end{subfigure} 
    \caption{Cactus plots comparing variants of \ToolName\ with different partitioning bits}
    \label{fig:cactus:inernal}
\end{figure*}

   \begin{figure*}[t]
	\centering 
    \includegraphics[width=\textwidth]{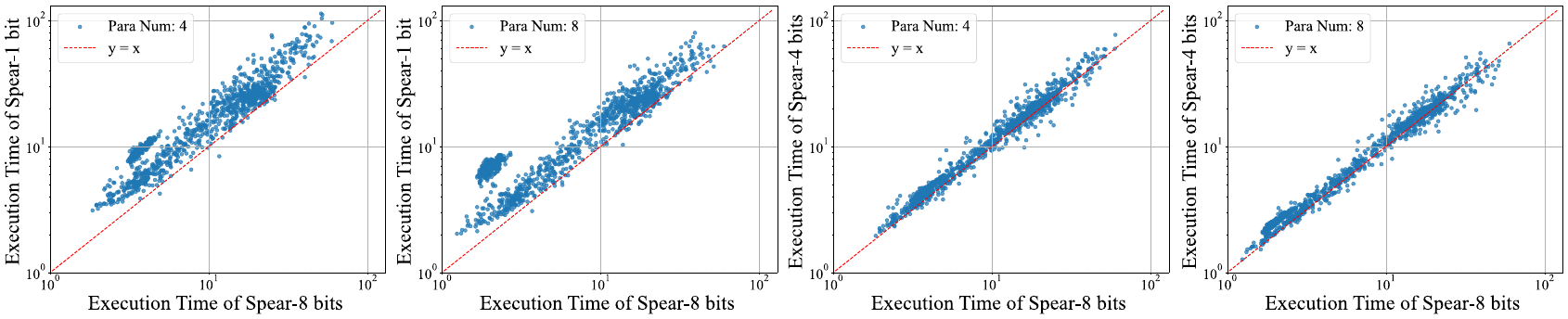}
    \caption{Comparing variants of \ToolName\ with different partitioning bits.}
    \label{fig:ssss}
\end{figure*}

%
%
%

We now evaluate variants of  Algorithm~\ref{algo:bit-level} by varying the number of bits processed per task (i.e., the size of partitioning bits). We consider six configurations: processing 1, 4, or 8 bits per parallel task across 4-core and 8-core environments. 
 \revise{
 Table~\ref{tbl:internal} reports the total time and average time per instance, while Figure~\ref{fig:cactus:inernal}-\ref{fig:ssss} summarizes the pairwise comparisons. Note that the total time represents the cumulative runtime across all $\sim$1,200 evaluated queries, rather than a single query.}
The number of solved instances for the configurations is nearly identical, so we omit the data from the table.

\smallskip
\noindent \textbf{Partitioning Granularity}.
Increasing the bits per task consistently improves performance. With 4 cores, \ToolName-8 bits achieves an 11.6s average time—10\% faster than \ToolName-4 bits (12.7s) and 60\% faster than \ToolName-1 bit (18.6s). This trend holds for 8 cores, where \ToolName-8 bits (10.1s) outperforms \ToolName-4 bits (10.6s) and \ToolName-1 bit (14.9s) by 5\% and 48\%, respectively. Larger bit partitions reduce synchronization overhead by consolidating related bit constraints into fewer tasks, reducing communication costs between cores. However, the marginal gain diminishes beyond 4 bits, suggesting diminishing returns at coarser granularities (i.e., larger $k$ for the partition bits)

\smallskip
\noindent \textbf{Core Scalability}.
Doubling the number of cores from 4 to 8 improves performance across all partitioning strategies. \ToolName-8 bits scales best, reducing the average time by 13\% (11.6s$\rightarrow$10.1s), compared to 17\% for \ToolName-4 bits (12.7s$\rightarrow$10.6s) and 20\% for \ToolName-1 bit (18.6s$\rightarrow$14.9s). Smaller bit partitions benefit more from additional cores due to increased task parallelism, but their higher baseline overhead limits absolute gains.

%% file: 6discuss.tex
\section{Discussions}
\label{sec:discuss}


\revise{This section focuses on the practical use of \ToolName. We discuss applicability and design choices, outline possible extensions, and summarize the main factors that limit or qualify the interpretation of our results.}

\subsection{Applicability and Design Choices}
\label{subsec:discuss:applicability}

\noindent \textbf{Applicability of OMT-based Solutions}.
Our algorithm generalizes naturally to more expressive abstract domains, such as octagons, by leveraging standard template-based invariant inference techniques. In this setting, unknown parameters are synthesized to satisfy semantic constraints. Besides, although our implementation targets bit-vector encodings, the underlying method is agnostic to the specific constraint theory. For instance, floating-point constraints can be handled via bit-vector encodings using bit-blasting, and finite-field arithmetic can be similarly reduced by encoding as bit-vectors. The speedup achieved by \ToolName\ can be sufficient to make previously infeasible analyses practical in industrial settings: when a static analysis exceeds typical time budgets (e.g., 5--10 hours per run~\cite{mcpeak2013scalable}), reducing solver time by several times can bring the pipeline within acceptable limits.

\smallskip
\noindent \textbf{Online vs.\ Offline Symbolic Abstraction}.
Symbolic abstraction can be used in two distinct ways, as discussed by Thakur and Reps~\cite{thakur2012method}. The first is to compute an \emph{explicit representation} of a transfer function $F$ (e.g., as a table or closed form), which is typically done \emph{offline} before the main analysis; the resulting $F$ is then applied repeatedly to abstract elements during analysis. This offline mode is often used to synthesize transfer functions for individual instructions or library functions, which are later composed during analysis. The second is to compute the abstraction of a \emph{specific} formula on demand---i.e., to evaluate $F(A)$ for a given abstract element $A$ without ever constructing a full representation of $F$. This online mode is closer in spirit to software model checking: solvers are used to \emph{apply} the (implicit) transfer functions of one or more blocks of code on demand, rather than precompute and store them. The two modes target different use cases. In this work, our primary setting (e.g., hybrid fuzzing with Angr) follows the second: we compute on-demand abstractions of path constraints for hard-to-cover branches, so each query corresponds to a single boxed OMT instance, and we do not materialize a global transfer function. \ToolName\ is designed to precisely accelerate this per-query, on-demand regime. That said, the \emph{offline} approach remains highly valuable when a reusable transformer is needed (e.g., for fixpoint iteration over a CFG). Synthesizing precise transformers offline for large-scale code can still produce challenging OMT instances; our algorithm can also accelerate that setting by parallelizing the optimization within each instance. Finally, \emph{outer} parallelism---e.g., solving OMT for different code snippets or program points in parallel---is orthogonal and complementary to our \emph{inner} bit-level parallelism within a single boxed OMT instance; the two can be combined in a full analysis pipeline.

\subsection{Extensions}
\label{subsec:discuss:extensions}

\noindent \textbf{Extending the Parallel Solvers}.
\revise{While \ToolName's divide-and-conquer strategy is effective, its performance can be further improved by incorporating portfolio-based techniques. Running multiple solver instances with varied SAT backends introduces algorithmic diversity, which can improve robustness and reduce tail latency. Another open challenge is refining the dynamic load-balancing mechanism to better handle uneven partition complexities. One promising direction is to make the number of partition bits adaptive, allowing the solver to respond to problem-specific structure at runtime. Incorporating lightweight static analysis (e.g., value-set or interval over-approximations of objectives) could also prune the search space for certain bits and reduce the number of SAT calls.}

\smallskip
\noindent \textbf{Sub-Optimal Transformers}.
\revise{While \ToolName\ prioritizes optimality, practical scenarios may tolerate sub-optimality for responsiveness. Integrating anytime MaxSAT solvers~\cite{nadel2024tt} could yield progressively refined solutions, allowing users to halt solving early for ``good enough'' abstractions. The precision of such approximations relative to conventional abstract transformers is difficult to characterize formally. Nevertheless, they can effectively produce a stream of candidate solutions that may be valuable in practice. We believe this direction is of independent interest from both practical and theoretical perspectives.}

\subsection{Limitations and Threats to Validity}
\label{subsec:discuss:limitations}
\label{subsec:discuss:threats}

\noindent \textbf{Limitations and When to Use \ToolName}.
\revise{\ToolName\ is most beneficial when the analysis generates many boxed OMT queries with multiple objectives and non-trivial constraint formulas; in such cases, bit-level parallelism yields substantial gains over sequential and objective-level strategies. For analyses that produce only a few objectives per query or very small formulas, the overhead of task scheduling and shared-memory synchronization may outweigh parallel gains---in those settings, a sequential backend or a small number of cores may be preferable. Additionally, all workers share the same constraint formula $\varphi$, so memory usage is dominated by a single bit-blasted instance; scaling to extremely large formulas may require out-of-core or distributed strategies that we do not address here.}

\smallskip
\noindent \textbf{Threats to Validity}.
\revise{Several factors may affect the generalizability of our results. \emph{Internal validity}: We compared \ToolName\ against parallelized configurations of $\nu$Z and OptiMathSAT (portfolio mode), since neither natively supports boxed OMT in a divide-and-conquer style; the baselines thus reflect our best effort to use them in a parallel setting. \emph{External validity}: Our benchmarks are derived from two binary analysis clients (Angr and Clearblue) and a fixed set of programs; OMT instances from other domains (e.g., software product lines or configuration synthesis) may exhibit different characteristics. Expanding the benchmark set to additional application domains would strengthen external validity.
\emph{Construct validity}: We used a 120-second timeout per query, which is reasonable for analysis pipelines that must process many queries; longer timeouts might change the relative ranking on a subset of hard instances.}

%% file: 7related.tex
\section{Related Work}
\label{sec:related}


\subsection{Program Analysis}

\noindent \textbf{Best Abstract Transformers}.
Symbolic abstraction~\cite{reps2004symbolic,DBLP:journals/entcs/ThakurLLR15} computes the best abstraction of a given formula $\varphi$ in a specific abstract domain.
The problem has been studied in various contexts, including heap-manipulation programs~\cite{reps2004symbolic,DBLP:conf/tacas/YorshRS04}, binaries~\cite{tsl,yao2021program}, numerical programs~\cite{jiang2017block}, and Markov decision processes~\cite{wachter2010best}.
Constructing the best abstract transformer has also been considered in the context of 
Indeed, symbolic abstraction is closely related to predicate abstraction~\cite{Graf97} in the model-checking community, in which the abstract element is a Boolean combination of a given set of predicates.
There are two categories of approaches to solving the symbolic abstraction problem.
One category of existing approaches uses SMT-based iterative refinement:
Some techniques work from ``below'', identifying a chain of successively weaker implicants until one is a consequence of $\varphi$~\cite{reps2004symbolic}; Other techniques work from ``above,'' identifying a chain of successively stronger implicates until no further strengthening is possible~\cite{thakur2012method,thakur2012bilateral}.
Another category reduces the problem to other automated reasoning tasks, such as OMT solving~\cite{yao2021program,li2014symbolic} and quantifier elimination~\cite{brauer2011transfer}.
However, neither category addresses the parallel synthesis of best abstract transformers.

	
\smallskip
\noindent \textbf{Analysis of Modular Arithmetic}.
Addressing the discrepancy between mathematical and finite-precision integers is crucial to the design of numeric domains. Consequently, there is extensive research on abstract domains for bit-vector arithmetic~\cite{gange15, mine:hal-00748094, 10.1007/978-3-319-52234-0_27, 10.1007/978-3-540-74061-2_8}. 
Astree~\cite{blanchet2002design} and cccheck~\cite{fahndrich2010static} 
	detect expressions guaranteed not to cause overflow or underflow while issuing warnings for potentially unsafe expressions. 
	The wrapped interval domain~\cite{gange15} models overflow and underflow by wrapping bit-vector values around their minimum and maximum limits. 
	These approaches rely on traditional, instruction-level abstract interpretation, which may not always yield optimal abstract transformers. In contrast, symbolic abstraction for modular arithmetic has also been applied to various domains, such as intervals~\cite{regehr2006deriving,brauer2011transfer}, sets~\cite{brauer2010automatic}, affine relations~\cite{elder2014abstract,tsl}, octagons~\cite{sharma2017sound},
	and polyhedra~\cite{sharma2017sound,yao2021program}. 
	Transfer functions for low-level code have been synthesized for intervals using BDDs by applying interval subdivision, where the extrema representing the interval are bit-vectors.
	We focus on template domains that use a finite set of linear constraints. 
	
	

\smallskip
\noindent \textbf{Parallel Program Analysis}. Parallel methods have also been employed in static analysis to enhance efficiency, such as pointer analysis~\cite{10.1145/1869459.1869495}, CFL reachability~\cite{6957254}, rapid type analysis (RTA)~\cite{10.1145/3106237.3106261}, and symbolic execution~\cite{palikareva2013multi}. There have also been a few efforts to parallelize abstract interpretation. \citet{10.1145/996841.996869} develop a tool called C Global Surveyor (CGS), which distributes the analysis across multiple processors in a cluster. 
Building on this, \citet{7054185} propose a novel approach to program analysis based on abstract interpretation, which splits the analysis into two components: a parallel exploration of the transition system and a separate state-space manager. 
Our work complements these efforts by enabling parallel synthesis of the best abstract transformers.

\subsection{Automated Reasoning}

\smallskip
\noindent \textbf{Optimization Modulo Theories}.
OMT~\cite{tsiskaridze2024generalized} has been studied extensively, with proposed solutions spanning various domains, including linear arithmetic~\cite{LIA-OPT-Masters}, pseudo-Boolean constraints~\cite{PB-OMT12,Sebastiani:2015:OMT:2737801.2699915,DBLP:journals/corr/SebastianiT17}, bit-vectors~\cite{bjorner2015nuz,nadel2016bit,DBLP:journals/corr/abs-1905-02838}, and nonlinear arithmetic~\cite{DBLP:conf/frocos/BigarellaCGIJRS21}. \revise{In program-analysis settings, OMT has also been used for symbolic optimization and transformer synthesis through SMT-based optimization, quantifier-instantiation-based reasoning, and iterative refinement~\cite{li2014symbolic,brauer2011transfer,kong2018delta-decision,yao2021program}. Boxed OMT generalizes this view to multiple separate objective functions over a shared formula, where modern solvers can share Boolean search, theory minimization, and improvement clauses while optimizing several objectives~\cite{sebastiani2015pushing,DBLP:journals/corr/SebastianiT17}.}
Existing algorithms can be broadly classified into two categories.
The \emph{offline schema} treats the SMT solver as a black box, performing the optimization search through incremental solver invocations \cite{PB-OMT12,Sebastiani:2015:OMT:2737801.2699915}. This schema employs linear or binary search strategies, iteratively tightening the bounds on the objective function after each solver call. In contrast, the \emph{inline schema} incorporates the optimization criteria directly into the internal workings of the SMT solver \cite{PB-OMT12, Sebastiani:2015:OMT:2737801.2699915}. \revise{The calculus-based view of generalized OMT is orthogonal to this classification: it provides a theory-agnostic framework whose rule-application strategies can model offline, inline, and hybrid OMT procedures~\cite{tsiskaridze2024generalized}.}
The inline schema necessitates significant modifications to the solver.
We follow the offline schema and rely on a SAT oracle.  
Our approach can be extended to variables with finite value sets, such as floating points.

\smallskip
\noindent \textbf{Parallel Constraint Solving}.
The two primary categories of parallel solving approaches are portfolio and divide-and-conquer. In the portfolio approach, multiple solvers are executed in parallel on the same input formula, and the solution is obtained from the first solver that determines satisfiability~\cite{xu2008satzilla, DBLP:conf/cav/WintersteigerHM09}. Some portfolio frameworks extend this strategy by enabling solvers to share information, such as learned lemmas~\cite{clauseSharingCloud}.
A significant portion of the research on SMT solving has concentrated on portfolio-based techniques. Z3 was among the first solvers to implement portfolio solving with information sharing~\cite{concurrentZ3}. Additionally, SMTS~\cite{SMTS} provides a parallel framework for portfolio solving that supports information sharing~\cite {clauseSharingCloud}.
In contrast, the divide-and-conquer approach partitions the search space of the input formula into distinct subspaces, which are then solved in parallel~\cite{DBLP:conf/sat/HyvarinenMS15,paropensmt,DBLP:conf/atva/MarescottiHS16}. 
\revise{Modern advances in this space explore highly sophisticated partitioning topographies for distributed SMT solving, which dynamically divide the search space by selecting optimal splitting variables based on internal solver heuristics~\cite{WilsonNRCTB23}.}
Another notable example of this approach is cube-and-conquer, which has been successfully applied to SAT~\cite{cubeAndConquer} and SMT~\cite{AnttiLookahead,pboolector}.
OpenSMT2 implements two lookahead strategies for generating cubes in a cube-and-conquer-like partitioning scheme~\cite{AnttiLookahead}. One approach is based on the global number of free atoms, while the other relies on the number of unassigned atoms in the clauses. PBoolector~\cite{pboolector} employs a cube-and-conquer strategy for solving QF\_BV formulas.
Our work is the first to address parallel solving for boxed OMT problems. We adopt a divide-and-conquer strategy, partitioning the search space along the bit-level structure of the objective functions.


%% file: 8conclusion.tex
\section{Conclusion}
\label{sec:conclu}
The construction of precise abstract transformers is a critical bottleneck in scaling abstract interpretation to industrial-scale programs.
This paper presented \ToolName, a fine-grained parallel algorithm for solving the problem. Experimental results demonstrate that \ToolName\ consistently outperforms existing solvers on benchmarks derived from two real-world binary analysis applications.

%% file: 9dataavail.tex
\section*{Data Availability}
We release the source code and datasets used in this work at \href{https://anonymous.4open.science/r/para_omt-4D86}{this repository}.